\documentclass[journal, onecolumn]{IEEEtran}

\usepackage{geometry} 
\usepackage{graphicx}
\usepackage{amssymb}
\usepackage{amsmath}
\usepackage{amsthm}
\usepackage{epstopdf}
\usepackage{bbm} 
\usepackage{array} 
\usepackage[lofdepth,lotdepth]{subfig}

\DeclareMathOperator{\erf}{erf}
\newtheorem{thm}{Theorem}
\newtheorem{lem}{Lemma}
\newtheorem{as}{Assumption}
\newtheorem{cor}{Corollary}
\theoremstyle{definition}
\newtheorem{defin}{Definition}
\ifCLASSINFOpdf
\else
\fi

\begin{document}
\title{Weighted $\ell_1$-minimization for generalized non-uniform sparse model}
%
%
%
%

\author{Sidhant~Misra and Pablo~A.~Parrilo
 \thanks{Laboratory of Information and Decision Systems, Department of Electrical Engineering and Computer Science, MIT, Cambridge, MA, 02139, } }
\maketitle

\begin{abstract}
Model-based compressed sensing refers to compressed sensing with extra structure about the underlying sparse signal known a priori. Recent work has demonstrated that both for deterministic and probabilistic models
imposed on the signal, this extra information can be successfully exploited to enhance recovery performance. In particular, weighted $\ell_1$-minimization with suitable choice of weights has been shown to improve
performance in the so called non-uniform sparse model of signals. In this paper, we consider a full generalization of the non-uniform sparse model with very mild assumptions.
We prove that when the measurements are obtained using a matrix with i.i.d Gaussian entries, weighted $\ell_1$-minimization successfully recovers the sparse signal from its measurements with overwhelming probability. We also provide a method to choose these weights for any general signal model from the non-uniform sparse class of signal models. \end{abstract}



%

\begin{section}{Introduction}
Compressed Sensing has emerged as a modern alternative to traditional undersampling for compressible signals. Previously, the most common way to view recovery of signals from samples was based on the Nyquist criterion.  According to the Nyquist criterion, a band-limited signal has to be sampled at a rate at least twice its bandwidth to allow exact recovery. In this case the band-limitedness
of the signal is the extra known information that allows us to reconstruct the signal from compressed measurements. In compressed sensing the additional structure considered is that the signal is sparse with respect to a certain known basis. As opposed to sampling at Nyquist rate and subsequently compressing, we now obtain linear measurements of the signal and the compression and measurement steps are now combined by obtaining much smaller number of linear measurements than what would be in general required to reconstruct the signal from its measurements. 

After fixing the basis with respect to which the signal is sparse, the process of obtaining the measurements can be written as
	$\mathbf{y} = \mathbf{Ax},$
where,
$\mathbf{y} \in \mathbbm{R}^m$ is the vector of measurements, $\mathbf{x} \in \mathbbm{R}^n$ is the signal and 
$\mathbf{A} \in  \mathbbm{R}^{m \times n}$ represents the $m$ linear functionals acting on the signal $\mathbf{x}$. We call $\mathbf{A}$ the \textit{measurement matrix}.
The signal $\mathbf{x}$ is considered to have at most $k$ non-zero components and we are typically interested in the scenario where $k$ is much smaller than $n$.
Compressed Sensing revolves around the fact that for sparse signals, the number of such linear measurements needed to reconstruct the signal can be significantly smaller than the ambient dimension of the signal itself. The reconstruction problem can be formulated as finding the sparsest solution $\mathbf{x}$ satisfying the constraints imposed by the linear measurements $\mathbf{y}$. This can be represented by
\begin{align*}
	\min \quad & ||\mathbf{x}||_0 \\
	\mbox{subject to} \quad &\mathbf{ y}=\mathbf{Ax}.
\end{align*}
This problem is inherently combinatorial in nature and is in general a NP-hard problem. This is because for a certain value of the size of the support of $\mathbf{x}$ given by $k$, one needs to search through all 
$n \choose k$ possible supports of the signal. Seminal work by Cand\'{e}s and Tao in \cite{CandesTao} and Donoho in \cite{Donoho1} show that under certain conditions on the measurement matrix $\mathbf{A}$, $\ell_1$ norm 
minimization, which can be recast as a linear program, can recover the signal from its measurements. Additionally, a random matrix with i.i.d. Gaussian entries with mean zero satisfies the required condition with overwhelming
probability. Linear programming is known to have polynomial time complexity and the above 
mentioned result tells us that for a large class of measurement matrices $\mathbf{A}$ we can solve an otherwise NP-hard combinatorial problem in polynomial time. Subsequently, iterative methods based on a greedy approach were formulated which 
recover sparse signals from their
 measurements by obtaining an increasingly accurate approximation to the actual signal in each iteration. Examples of these include CoSaMP \cite{NeedellTropp} and IHT \cite{BlumensathDavis}.
The compressed sensing framework has also been generalized to the problem of recovering low rank matrices from compressed linear measurements represented by linear matrix equations \cite{RechtFazelParrilo10}
\cite{RechtXuHassibi10}.

Most of the earlier literature on Compressed Sensing focused on the case where the only constraints on the signal $\mathbf{x}$ are those imposed by its measurements $\mathbf{y}$. On the other hand it is natural to consider the case where besides sparsity, there is certain additional information on the structure of the underlying signal. This is true in several applications, examples of which include encoding natural images (JPEG), MRI and DNA 
microarrays \cite{BaraniukRosing07}, \cite{Sidhant08}. This leads us to Model Based Compressed Sensing where the aim is to devise recovery methods specific to the signal model at hand. Furthermore, one would also want to quantify the possible benefits it has over the standard method (e.g. lesser number of required measurements for the same level of sparsity of the signal). This has been explored in some recent papers. The authors in \cite{Baraniuk1} analyzed a deterministic signal model, were the support of the underlying signal is constrained to belong to a given known set. This defines a subset $\mathcal{M}$ of the set of all $k$-sparse signals, which is now the set of allowable signals. This results in an additional constraint on the original reconstruction problem.
\begin{align*}
	\min \quad &||\mathbf{x}||_0 \\
	\mbox{subject to} \quad &\mathbf{y} = \mathbf{Ax}, \\
	& \mathbf{x} \in \mathcal{M}. 
\end{align*}
It was shown that an intuitive modification to the CoSaMP or IHT method succeeds in suitably exploiting the information about the model. The key property defined in \cite{CandesTao}, known as the Restricted Isometry Property was adapted in \cite{Baraniuk1} to a model based setting. With this, it was shown that results similar to \cite{CandesTao} can be obtained for model-based signal recovery. 

As opposed to this, a probabilistic model, i.e. a Bayesian setting, was considered in \cite{WeiyuThesis}. Under this model there are certain known probabilities associated with the components of the signal $\mathbf{x}$. Specifically, 
$p_i , i = 1,2, \ldots , n$ with $0 \leq p_i \leq1$ are such that
\begin{align*}
	\mathbf {P} (x_i \mbox{  is non-zero}) = p_i \qquad  i = 1,2, \ldots , n.
\end{align*}
The deterministic version of the same model called the ``nonuniform sparse model" was considered in \cite{XuHassibi10}. 
For this, the use of weighted $\ell_1$-minimization  was suggested, given by
\begin{align*}
	\min \quad &||\mathbf{x}||_{w,1} \\
	\mbox{subject to} \quad &\mathbf{y} = \mathbf{Ax},
\end{align*}	
where $ ||\mathbf{x}||_{w,1} = \sum_{i=1}^{n} w_i |x_i|$ denotes the weighted $\ell_1$ norm of $\mathbf{x}$. The quantities $w_i, i=1,2, \ldots, n$ are some positive scalars. Similar to  \cite{Donoho1}, \cite{DonohoTanner05}, ideas based on high dimensional polytope geometry were used to provide sufficient conditions under which weighted $\ell_1$-minimization recovers the sparse signal. This method was 
 introduced earlier in \cite{XuHassibi08} where it was used to analyze the robustness of $\ell_1$-minimization in recovering sparse signals from noisy measurements.
 The specific model considered in \cite{WeiyuThesis} can be described as 
 follows. Consider a partition of the indices $1$ to $n$ into two 
 disjoint sets $T_1$ and $T_2$. Let $p_i = P_1, i\in T_1$ and $p_i = P_2, i \in T_2$. As a natural choice choose the weights in the weighted $\ell_1$-minimization as $w_i = W_1, i \in T_1$ and $w_i = W_2, i \in T_2$. The main
 result of \cite{WeiyuThesis} is that under certain conditions on $\frac{W_2}{W_1}$, $\frac{P_2}{P_1}$ and $\mathbf{A}$, weighted $\ell_1$-minimization can recover a signal drawn from the above model with overwhelming probability. This was later extended in \cite{XuHassibi10} to the case when the support of the sparse signal is divided into $u$ classes, where $u$ is a fixed constant, see Theorem 5.3 in \cite{XuHassibi10}.
 
 In this paper, we consider a generalization of the non-uniform sparse model with mild assumptions. In particular, we keep the assumption that each element of the sparse signal has a probability $p_i$ of being non-zero 
 independent of every other element. To describe the behavior as $n \rightarrow \infty$, we assume that these probabilities converge to some shape function $p(.)$ i.e.,
 \begin{align}
 	\lim_{n \rightarrow \infty} p_{un}  = p(u), \quad 0\leq u \leq 1.
 \end{align}

 \begin{as} \label{as:p}
  	The probability shape function $p(.)$ satisfies the following conditions:

 \begin{itemize}
 	\item[(a)] $p(.)$ is continuous in $[0,1]$ except at finitely many points.
	\item[(b)] Within each interval that $p(.)$ is continuous, it is monotonically non-increasing.
 \end{itemize}
 \end{as}
The above assumptions are mild and many shape functions with finitely many discontinuities can be ``rearranged" to satisfy the second assumption.


The possibility of specifying a non-uniform probability shape function can be useful in several applications. For example, natural images often have sparsity properties after certain transforms, including the Fourier and wavelet transforms. Furthermore, they tend to have significantly more content in the lower frequencies, which causes a natural decay in the likelihood of non-zero Fourier coefficients with increasing frequencies. This decay can then be modeled as the shape function $p(\cdot)$ described above. In the context of specific applications, this information can be deduced, for instance, from the statistics of previous training data.

For a signal drawn from this model, we propose the use of weighted $\ell_1$-minimization to reconstruct it from its compressed linear measurements, and in addition, to match the signal model, we propose that the weights be chosen according to a  shape function $f(.)$ which satisfies the same properties as $p(.)$, except that it is monotonically non-decreasing instead.
Note that, we have still kept the independence assumption, which does not take into account for example correlations in the signal support. However, weighted $\ell_1$-minimization is generally not a suitable choice in that case which is why we do not consider such correlations in our model.
 
  We prove that under certain conditions on $p(.)$, $f(.)$ and the measurement matrix $\mathbf{A}$, we can reconstruct the signal perfectly with overwhelming probability. 
 Although a good part of the machinery in \cite{Donoho1} and \cite{XuHassibi10} applies to our case, difficulties arise because of the generality of the class of functions described above.
 In addition to computing the usual angle exponents, typical of the high-dimensional geometry based analysis in \cite{Donoho1}, we also need to investigate further properties of these exponents with respect to the weight function
 $f(.)$, such as monotonicity, in order to compute a true upper bound on the probability of failure. Additionally, to make sure that the angle exponents, which are described as the solution to an optimization problem, can be efficiently computed numerically, we prove their 
 partial convexity with respect to the optimizing variables. { The angle exponent analysis along with the monotonicity and partial convexity properties help us formulate and prove sufficient conditions for sparse signal recovery that also allow efficient numerical verification. }
 
 A major question that arises in our general setting, is how to choose the function $f(.)$ for the weighted $\ell_1$-minimization according to $p(.)$. We answer this question by
providing a definitive recommendation for the choice of weights based on Gaussian width computations using Gordon's theorem. This, along with our angle exponent analysis, which is generally known to be quite tight, provides a general framework for designing a weighted $\ell_1$-minimization based recovery method as well as
 a way to verify its performance theoretically for the non-uniform sparse model.
 
 The rest of the paper is organized as follows.
 In Section \ref{sec:problem_formulation}, we introduce the basic notation. We formulate the exact questions that we set out to answer in this paper and state our main theorem. In Section \ref{sec:analysisofl1} we focus on how
 weighted $\ell_1$-minimization behaves by restricting our attention to special class of signals that are particularly suitable for the ease of analysis. We later show that the methods generalize to other simple classes and is
 all we need to establish the main result of this paper. In Section \ref{sec:withprob} we describe the key features of a typical signal drawn from our model. We also prove a suitable large deviation result for the probability that a signal drawn from our model lies outside this typical set. In Section \ref{sec:choosing}, we provide a method for deriving a recommended weight function $f(.)$ based on minimizing an upper bound on Gaussian width. In Section \ref{sec:simulation} we provide numerical computations to demonstrate the results we derive and then provide simulation results. We  conclude the paper in Section \ref{sec:conclusion}.
%
%
%
%
\end{section}

\begin{section}{Problem Formulation} \label{sec:problem_formulation}

	\begin{subsection}{Notation and parameters}
		We denote scalars by lower case letters (e.g. c), vectors by bold lower case letters (e.g. \textbf{x}), matrices with bold upper case letters (e.g. \textbf{A}). Probability of an event \textit{E} 
		is denoted by $\mathbf{P}(E)$. The $i^{th}$ standard unit vector in $\mathbbm{R}^n$ is denoted by $e_i = (0, 0 , \ldots , 1, \ldots, 0)^T$, where the ``$1$" is located in the $i^{th}$ position.
		
		The underlying sparse signal is represented by $\mathbf{x} \in \mathbbm{R}^n$, the measurement matrix by $\mathbf{A} \in \mathbbm{R}^{m \times n}$. The vector of observations is 
		denoted by $\mathbf{y}$ and is obtained through linear measurements of $\mathbf{x}$ given by $\mathbf{y} = \mathbf{Ax}$. Typically we would need $n$ linear measurements to be able to recover
		the signal. The scalar $\alpha = \frac{m}{n}$ determines how many measurements we have as a fraction of $n$. We call this the compression ratio.
	\end{subsection}
		
	\begin{subsection}{Model of the Sparse Signal} \label{sec:signalmodel}
		Let $p:[0,1] \rightarrow [0,1]$ be a continuous monotonically non-increasing function. We call $p(.)$ the \textit{probability shape function} and the reason for this name will become clear from the description below.
		The support of the signal $\mathbf{x}$ is decided by the outcome of $n$ independent Bernoulli 
		random variables. In particular, if $E_i$ denotes the event that $i \in Supp(\mathbf{x})$, then the events $E_i \quad i=1, \ldots ,n$ are independent and $\mathbf{P}(E_i) = p\left( \frac{i}
		{n}\right)$. Although we assume throughout the paper that $p(.)$ is continuous, our result generalizes to piecewise continuous functions as 
		well. Let us denote by $k$ the cardinality of $Supp(\mathbf{x})$. Note that under the above model $k$ is a random variable that can take any value from $0$ to $n$. The expected value of $k$
		is given by $\sum_{i=1}^{n} p\left( \frac{i}{n} \right)$. We denote by $\delta$ the expected fractional sparsity given by $\delta = \frac{1}{n} \mathbf{E}[k] = \frac{1}{n} \sum_{i=1}^{n} p \left( \frac{i}{n}
		\right)$.
		
		We notice here that the signal model described above is much more restrictive than the ones made in Assumption \ref{as:p}. However, as it turns out, the exact same analysis also works for the more general case
		in Assumption \ref{as:p}, so we stick to the simplified model above to keep the notation and analysis less cumbersome.
		
		As is standard in Compressed Sensing literature we assume that the entries of the measurement matrix $\mathbf{A}$ are i.i.d. Gaussian random variables with mean zero and variance 
		one. The measurements are obtained as $\mathbf{y} = \mathbf{Ax}$.
	\end{subsection}
	
	\begin{subsection}{Weighted $\ell_1$-minimization}
		When the fractional sparsity $\delta$ is much smaller than one, a signal sampled from the model described above is sparse. Hence, it is possible to recover the signal from its measurements $\mathbf{y}$
		by $\ell_1$-minimization which is formulated as
		\begin{align*}
			\min \quad & ||\mathbf{x}||_1 \\
			\mbox{subject to} \quad & \mathbf{Ax} = \mathbf{y}.
		\end{align*}
		However this does not exploit the extra information available from the knowledge of the priors. Instead, we use weighted $\ell_1$-minimization to recover the sparse signal which is 
		captured by the following optimization problem:
		\begin{align}
			\min \quad & ||\mathbf{x}||_{\mathbf{w},1} \label{eq:weightedl1}\\ 
			\mbox{subject to} \quad & \mathbf{Ax} = \mathbf{y}. \notag
		\end{align}
		where $\mathbf{w} \in \mathbbm{R}^n$ is a vector of positive weights and $||x||_{\mathbf{w},1} = \sum_{i=1}^{n} w_i |x_i|$ refers to the weighted $\ell_1$ norm of $\mathbf{x}$, for a 
		given weight vector $\mathbf{w}$. The weight vector $\mathbf{w}$ plays a central role in determining whether (\ref{eq:weightedl1}) successfully recovers the sparse 
		signal $\mathbf{x}$. Intuitively, $\mathbf{w}$ should be chosen in a certain way depending on $p(.)$ so as to obtain the best performance (although at this point we have not precisely 
		defined the meaning of this). Keeping in mind the structure of $p_i, \  i = 1,\ldots,n$, we suggest using weights $w_i, \ i = 1, \ldots, n$ which have a similar structure.
		Formally, let $f: [0,1] \rightarrow \mathbbm{R}^+$ be a non-negative non-decreasing continuous function. Then we choose the weights as $w_i = f \left( \frac{i}{n} \right)$. We call $f(.)$ the 
		\textit{weights shape function}.
	
	\end{subsection}
	
	\begin{subsection}{Problem statement}
		In this paper we try to answer the following two questions:
		\begin{itemize}
		\item Given the problem parameters (size of the matrix defined by $m,n$), the functions $p(.)$ and $f(.)$, does weighted $\ell_1$-minimization in (\ref{eq:weightedl1})
		recover the underlying sparse signal $\mathbf{x}$ with high probability?
		\item Given a family of probability shape functions $p(.: \delta)$, how to choose the weight function $f(.)$ that has the best
		performance guarantees? Here $\delta = \int_{0}^{1} p(u; \delta) du$ is the expected sparsity of the model.
		\end{itemize}
		
		We give an answer in the affirmative to the first question, given that the functions $p(.)$ and $f(.) $ satisfy certain specified conditions. 
		This is contained in the main result of this paper which is 
		\begin{thm} \label{thm:mainresult}
			Let the probability shape function $p(.)$ and the weight shape function $f(.)$ be given. Let E be the event that weighted $\ell_1$-minimization described in (\ref{eq:weightedl1})
			 fails to recover the correct sparse vector $\mathbf{x}$.
			There exists a quantity $\bar{\psi}_{tot}(p,f)$ which can be computed explicitly as described in Appendix \ref{completetotalexponent},
			 such that whenever $\bar{\psi}_{tot}(p,f)<0$ the probability of failure $\mathbf{P}(E)$
			of weighted $\ell_1$-minimization decays exponentially with respect to $n$. More precisely, if for some  $\epsilon>0$ we have $\psi_{tot}(p,f) \leq -\epsilon$.
			then there exists a constant $c(\epsilon) >0$ such that for large enough $n$, the probability of failure satisfies $\mathbf{P}(E) \leq e^{-n c(\epsilon)}$.
		\end{thm}

		To answer the second question, we need to define 
		 the measure of performance. For a 
		given value of $\alpha = \frac{m}{n}$, let
		\begin{align}
			\bar{\delta} \quad \triangleq \quad \quad \max \quad & \delta \\  \label{deltabaroriginal}
			\mbox{subject to} \quad & \bar{\psi}_{tot}(p(.; \delta), f(.)) \leq 0, 
		\end{align}
		We call $\bar{\delta}$ the guaranteed bound on recoverable sparsity $\delta$. In Section \ref{sec:choosing}, we answer the question of choosing the best $f(.)$, and in Section \ref{sec:simulation} we use numerical 
		computations to demonstrate that this choice indeed outperforms other natural choices when we compare the $\bar{\delta}$ attained by each choice.
	\end{subsection}
	
\end{section}

\begin{section}{Analysis of Weighted $\ell_1$-minimization} \label{sec:analysisofl1}

	In this section, we analyze the performance of weighted $\ell_1$-minimization with weights specified by $f(.)$ on a certain special 
	class of sparse signals $\mathbf{x}$. As we will see in Section \ref{sec:withprob}, we can easily generalize the analysis to the signals 
	drawn from our signal model. 	
	
	Recall that the failure event $E$ from Theorem \ref{thm:mainresult} is defined as
	 the event that weighted $\ell_1$-minimization fails to recover the correct sparse vector $\mathbf{x}$. We call the probability of this event $\mathbf{P}(E)$
	as the probability of failure. Without loss of generality we can assume that 
	$||\mathbf{x}||_{\mathbf{w},1} = 1$, that is $\mathbf{x}$ lies on a $k-1$-dimensional face of the weighted cross-polytope
	\begin{align*}
		\mathcal{P} \triangleq \{ x \in \mathbbm{R}^n \ s.t. \ ||\mathbf{x}||_{\mathbf{w},1} \leq 1\}
	\end{align*}
	which is the weighted $\ell_1$-ball in $n$ dimensions. The specific face of  $\mathcal{P}$ on which $\mathbf{x}$ lies is determined by the support of $\mathbf{x}$. The probability of failure $\mathbf{P}(E)$
	 can be written as 
	\begin{align*}
		\mathbf{P}(E) = \sum_{F \in \mathcal{F}} \mathbf{P}(\mathbf{x} \in F) \mathbf{P}(E|\mathbf{x} \in F),
	\end{align*}
	where $\mathcal{F}$ is the set of all faces of the polytope $\mathcal{P}$.
	Then, as shown in \cite{WeiyuThesis}, the event $\{E| \mathbf{x} \in F\}$ is precisely the event that there exists a non-zero $u \in null(\mathbf{A})$ such that 
	$||\mathbf{x + u}||_{\mathbf{w},1} \leq ||\mathbf{x}||_{\mathbf{w},1}$ conditioned on $\{  \mathbf{x} \in F\}$.
	 Also since $\mathbf{A}$ has i.i.d. Gaussian entries, sampling from the null space of $\mathbf{A}$ is equivalent to sampling a subspace uniformly from the Grassmann
	manifold $Gr_{(n-m)}(n)$. So conditioned on $\{\mathbf{x} \in F\}$ the event $\{E|\mathbf{x} \in F\}$ is same as 
	the event that a uniformly chosen $(n-m)$-dimensional subspace shifted to $\mathbf{x}$ intersects the polytope $\mathcal{P}$ at a point other 
	than $\mathbf{x}$. The probability of the above event is also called the complementary Grassmann Angle for the face $F$ with respect to the polytope $\mathcal{P}$ under the Grassmann manifold
	$Gr_{(n-m)}(n)$. Based on work by Santal\'{o} \cite{San52} and McMullen \cite{McM75} the Complementary Grassmann Angle can be expressed explicitly as the sum of products of 
	 internal and external angles.
	\begin{align} \label{eq:decomposition}
		\mathbf{P}(E|\mathbf{x} \in F) = 2 \sum_{s \geq 0 }^{} \sum_{G \in J(m+1+2s)}^{} \beta(F,G) \gamma(G,\mathcal{P})
	\end{align}
	where $\beta(F,G)$ and $\gamma(G,\mathcal{P})$ are the internal and external angles and $J(r)$ is the set of all $r$-dimensional faces of $\mathcal{P}$. The definitions of internal and external angles can be found in
	 \cite{WeiyuThesis}). We include them here for completeness.
	 \begin{itemize}
		\item	 The \emph{internal angle} $\beta(F,G)$ is the fraction of the volume of the unit ball covered by the cone obtained by observing the
		 face $G$ from any any point in face $F$. The quantity $\beta(F,G)$ is defined to be zero if 
		 $F$ is not contained in $G$ and is defined to be one if $F = G$. In Figure \ref{fig:intext}, if the face $F$ is the vertex $C$ and the face $G$ is the face formed by $ABC$, then $\beta(F,G) = \beta$ as labelled.
	 	\item The \emph{external angle} $\gamma(G,\mathcal{\mathcal{P}})$ is defined to be the fraction of the volume of the unit ball covered by the cone formed by the outward normals to the hyperplanes supporting
	 	$\mathcal{P}$ at the face $G$. If $G  = \mathcal{P}$ then $\gamma(G,\mathcal{P})$ is defined to be one. {In Figure \ref{fig:intext}, if $G$ is the edge $BD$, then the external angle $\gamma(G, \mathcal{P}) = \gamma$
		as labelled. The ray $r_1$ is orthogonal to the face $BDC$ and the ray $r_2$ is orthogonal to the face $BDA$, i.e., normals to the supporting hyperplanes of $G =  BD$.}
	 \end{itemize}
	
	\begin{figure}
		\begin{center}
		\includegraphics[scale = 0.6]{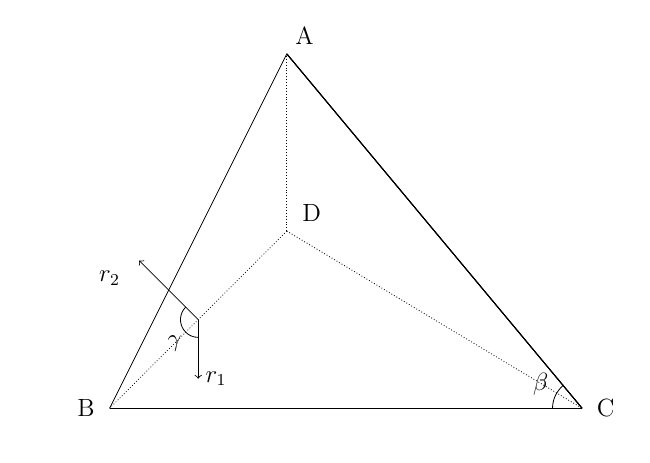}
		\end{center}
		\caption{Internal and External angles of faces of a polytope.}
		\label{fig:intext}
	\end{figure}
	
	In this section we describe a method to obtain upper bounds for  $\mathbf{P}(E|\mathbf{x} \in F)$ by finding upper bounds on the internal and external angles described above.
	We first analyze $\mathbf{P}(E | \mathbf{x} \in F)$ for the ``simplest" class of faces. We denote by 
	${F_0^k}$, the face whose vertices are given by $\frac{1}{w_1} e_1, \frac{1}{w_2} e_2, \ldots , \frac{1}{w_k} e_k$. Thus the vertices are defined by the first $k$ indices and
	we call this face the ``leading" $k-1$-dimensional face of $\mathcal{P}$. 
	We will spend much of this section developing bounds for $\mathbf{P}(E|\mathbf{x} \in F_0^k)$ for such ``leading" faces.
	 Then in Section \ref{sec:withprob}, we describe the \textit{typical set} of our signal model
	and show that for the purposes of bounding $\mathbf{P}(E)$ it is sufficient to consider a certain special class of faces. The bounds we develop in this section for ``leading" faces can be
	 easily generalized to faces belonging to this special class.
	
	\begin{subsection}{Angle Exponents for leading faces} \label{sec:firstface}
		We define the family of leading faces $\mathcal{F}_1$ as the set of all faces whose vertices are given by  $\frac{1}{w_1} e_1, \frac{1}{w_2} e_2, \ldots , \frac{1}{w_k} e_k$ for some $k$. 
		In this section we will establish the following result which bounds the internal and external angles related to a leading face of $\mathcal{P}$. We will then use this result in the next section 
		to provide an upper bound
		on $\mathbf{P}(E | \mathbf{x} \in F)$ for $F \in \mathcal{F}_1$.
		\begin{thm} \label{thm:firstfaceexponents}
			Let $\frac{k}{n} = \delta$ and $\frac{l}{n} = \tau$ with $\tau > \delta$. 
			Let $G_0^{l}$ be the face whose vertices are given by
			$\frac{1}{w_1}e_1, \frac{1}{w_2}e_2, \ldots , \frac{1}{w_l}e_l$ and
			let $F_0^{k} \subset G_0^{l}$ be the face whose vertices are given by $\frac{1}{w_1}e_1, \frac{1}{w_2}e_2, \ldots , \frac{1}{w_k}e_k$.
			Then there exist quantities
			 $\bar\psi_{ext}(\tau)$ and $\bar\psi_{int}(\delta,\tau)$ which can be computed explicitly as described in Section \ref{sec:firstfaceexternalangle} and
			 Section \ref{sec:firstfaceinternalangle} respectively, such that for any $\epsilon > 0$, there exist integers $n_0(\epsilon)$ and $n_1(\epsilon)$ satisfying  
			\begin{itemize}
				\item $n^{-1}\log(\beta(F_0^{k},G_0^{l})) <  \bar{\psi}_{int}(\delta,\tau) + \epsilon$, for all $n > n_0(\epsilon)$,
				\item $n^{-1}\log(\gamma(G_0^{l},\mathcal{P})) <  \bar{\psi}_{ext}(\tau) + \epsilon$, for all $n > n_1(\epsilon)$.
			\end{itemize}
			The quantities  $\bar{\psi}_{int}(\delta,\tau)$ and $\bar{\psi}_{ext}(\tau)$ are called the optimized internal angle exponent and the external angle exponent respectively.
		\end{thm}
		
		Whenever any of the angle exponents described above is negative, the corresponding angle
		decays at an exponential rate with respect to the ambient dimension $n$. We will prove Theorem \ref{thm:firstfaceexponents} over the following two 
		subsections by finding the optimized internal and external angle exponents that satisfy the requirements of the theorem. \\
		
		\begin{subsubsection}{\textbf{Internal Angle Exponent}}
			In this subsection we find the optimized internal angle exponent $\bar{\psi}_{int}(\delta,\tau)$ that satisfies the conditions described in Theorem \ref{thm:firstfaceexponents}.	
			We begin by  stating the following result from \cite{WeiyuThesis} which provides the expression for the internal angle of a face $F$ with respect to a face $G$ in
			 terms of the weights associated with their vertices. Subsequently, we find an asymptotic upper bound on the exponent of the internal angle using that expression. In what follows, we denote by
			 $HN(0,\sigma^2)$ the distribution of a half normal random variable obtained by taking the absolute value of a $N(0,\sigma^2)$ distributed random variable.
			\begin{lem} \cite{WeiyuThesis} \label{lem:internalexpression}
				Define $\sigma_{i,j} = \sum_{p=i}^{j}w_p^2.$
				Let $Y_0 \sim N(0,\frac{1}{2})$ be a normal random variable and $Y_p \sim HN(0,\frac{w_{p+k}^2}{2\sigma_{1,k}})$ for $p = 1,2,\ldots ,l-k$ be independent half normal distributed random 
				variables that are independent of $Y_0$. Define the random variable $Z \triangleq Y_0 - \sum_{p=1}^{l-k} Y_p$. Then, 
				\begin{align}
					\beta(F_0^{k},G_0^{l}) &  = \frac{\sqrt{\pi}}{2^{l-k}} \sqrt{  \frac{ \sigma_{1,l}  }{\sigma_{1,k}   }  }p_{Z}(0),  \notag
				\end{align}
				where $p_Z(.)$ is the density function of $Z$.
		\end{lem}
		
		We now proceed to derive an upper bound for the quantity $p_Z(0)$. Much of the analysis is along the lines of \cite{Donoho1}. Let the random variable $S$ be defined as $S = \sum_{p=1}^{l-k} Y_p$. 
		So, $Z = Y_0 - S$ and using the convolution integral, the density function of $Z$ can be written as
		\begin{align*}
			p_{Z}(0) &= \int_{-\infty}^{\infty}p_{Y_0}(-v)p_S(v)dv\\
			 &= 2\int_{0}^{\infty}vp_{Y_0}(v)F_S(v)dv,
		\end{align*}
		where $F_S(v)$ is the cumulative distribution function of $S$.
		Let $\mu_S$ be the mean of the random variable $S$. Then,
		\begin{align*}
			p_Z(0) =& \underbrace{2 \int_{0}^{\mu_S}vp_{Y_0}(v)F_S(v)dv}_{I} \\
			&+ \underbrace{ 2 \int_{\mu_S}^{\infty}vp_{Y_0}(v)F_S(v)dv}_{II}.
		\end{align*} 
		As in \cite{Donoho1}, the second term satisfies $II <  \int_{\mu_S}^{\infty} 2vp_{Y_0}(v)dv= e^{-\mu_S^2}$. As we will see later in this section $\mu_S^2 \sim cn$ for some $c>0$ and hence $II \sim e^{-c^2 n^2}$.
		Since we are interested in computing the asymptotic exponent of $p_Z(0)$, we can ignore $II$ from this computation.
		To bound $I$, we use $F_S(v) \leq exp(-\lambda_S^*(v))$, where $\lambda_S^*(v)$ denotes the rate function (convex conjugate of the characteristic function $\lambda_S(.)$) of the
		 random variable $S$. So we get
		\begin{align} \label{internalcramer}
			I & \leq \frac{2}{\sqrt{\pi}} \int_{0}^{\mu_S}v e^{-v^2 - \lambda_S^*(v)}dv.
		\end{align}
		For ease of notation we define the following quantities.
		\begin{align*}
			s_{k+1,l} &\triangleq \sum_{p=k+1}^{l} w_p \\
			\lambda(s) &\triangleq \frac{1}{2}s^2 + \log (2 \Phi(s)) \\
			\lambda_0(s) &\triangleq \frac{1}{s_{k+1,l}} \sum_{p=1}^{l-k}\lambda(w_{p+k}s) \\
			\lambda_0^*(y) &\triangleq \max_s sy - \lambda_0(s).
		\end{align*}
		Here $\lambda(s)$ is the characteristic function of the standard half normal random variable and $\Phi(x)$ is the standard normal distribution function. 
		Using the above definition, we can express the relation between $\lambda_S^{*}(y)$ and $\lambda_0^*(y)$ as
		\begin{align*}
			\lambda_0^*(y) = \frac{1}{s_{k+1,l}} \lambda_S^*\left(\frac{s_{k+1,l}}{\sqrt{\sigma_{1k}}}y \right)  
		\end{align*}
		We compute 
		\begin{align*}
			\mu_S = \sum_{p=1}^{l-k} \mathbf{E} Y_p = \sqrt{\frac{2}{\pi}} \frac{s_{k+1,l}}{\sqrt{\sigma_{1,k}}}
		\end{align*}
		Changing variables in (\ref{internalcramer}) by substituting $v =  \frac{s_{k+1,l}}{\sqrt{\sigma_{1,k}}} y$, we get $I \leq $
		\begin{align*}
			\frac{s_{k+1,l}^2}{\sqrt{\pi} \sigma_{1,k}} \int_{0}^{\sqrt{  \frac{2}{\pi} }} y  \exp[ - \frac{s_{k+1,l}^2}{2 \sigma_{1,k}} y^2 -  s_{k+1,l}\lambda_0^*(y) ]dy.
		\end{align*}
		Now, as $w_i = f\left( \frac{i}{n}\right) $, we have
		\begin{align*}
			s_{k+1,l} &= \sum_{i=k+1}^{l}f\left(\frac{i}{n}\right)  = n \left( \int_{\delta}^{\tau}f(x)dx + o(1) \right).
		\end{align*}
		Define $c_0(\delta,\tau) \triangleq \int_{\delta}^{\tau}f(x)dx$. This gives us $s_{k+1,l} = n(c_0(\delta,\tau) + o(1))$.
		Similarly,
		\begin{align*}
			\sigma_{1,k} &= \sum_{i=1}^{k}w_i^2 = \sum_{i=1}^{k}f^2\left(\frac{i}{n}\right) \\
			 &= n(\int_{0}^{\delta}f^2(x)dx+o(1)) \\
			 & =  n(c_1(\delta)+o(1)),
		\end{align*}
		where $c_1(\delta)$ is defined as $c_1(\delta) \triangleq \int_{0}^{\delta}f^2(x)dx$.
		This gives us
		\begin{align*}
			 \frac{s_{k+1,l}^2}{2\sigma_{k+1,l}}y^2 &+ s_{k+1,l}\lambda_0^*(y) \\
			 &=  n \left(        \frac{c_0^2}{2c_1}y^2 + c_0 \lambda^*(y)     +o(1) \right) \\
			 &=  n ( \eta(y) + o(1)).
		\end{align*}
		where $\eta(y)$ is defined as $\eta(y) \triangleq \left(        \frac{c_0^2}{2c_1}y^2 + c_0 \lambda_0^*(y)     \right)$.
		Using Laplace's method to bound $I$ as in \cite{Donoho1}, we get
		\begin{align} \label{eq:internalld}
		 	I &\leq R_n e^{-n \eta(y^*)},
		\end{align}
		where $n^{-1}\log(R_n) = o(1),$
		and $y^*$ is the minimizer of the convex function $\eta(y)$.
		Let
		\begin{align*}
			\lambda_0^*(y^*) &= \mbox{max}_s \mbox{ } sy^* - \lambda_0(s) = s^*y^* - \lambda_0(s^*).
		\end{align*}
		Then the maximizing $s^*$ satisfies $\lambda_0^{'}(s^*) = y^*.$
		From convex duality, we have $\lambda_0^{*'}(y^*) = s^*$.
		The minimizing $y^*$ of $\eta(y)$ satisfies
		\begin{align}
			&\frac{c_0^2}{c_1}y^* + c_0 \lambda_0^{*'}(y^*) = 0 \notag \\
			\implies & \frac{c_0^2}{c_1}y^* + c_0s^* = 0. \label{eq:yands}
		\end{align}
		This gives
		\begin{align} \label{eq:sstar}
			\lambda_0^{'}(s^*) = -\frac{c_1}{c_0}s^*.
		\end{align}
		First we approximate $\lambda_0(s)$ as follows
		\begin{align*}
			\lambda_0(s) &= \frac{1}{s_{k+1,l}}\sum_{p=1}^{l-k}\lambda(w_{p+k}s)  \\
			&=  \frac{1}{n c_0}\sum_{p=k+1}^{l}\lambda\left(f\left(\frac{p}{n}\right)s\right) + o(1) \\
			& =  \frac{1}{c_0}\int_{\delta}^{\tau}\lambda\left(sf(x)\right)dx + o(1).
		\end{align*}
		From the above we obtain
		\begin{align*}
			\frac{d}{ds}\lambda_0(s) =   \frac{1}{c_0}\int_{\delta}^{\tau}f(x)\lambda^{'}\left(sf(x)\right)dx + o(1).
		\end{align*}
		Combining this with equation (\ref{eq:sstar}), we can determine $s^*$ up to a $o(1)$ error by finding the solution to the equation
		\begin{align} 
			\int_{\delta}^{\tau}f(x)\lambda^{'}(s^*f(x))dx + c_1s^* = 0 \label{eq:findsstar}.
		\end{align} 
		We define the \textit{internal angle exponent} as
		\begin{align}
			{\psi}_{int}(\beta,\tau,y) & \triangleq -(\tau  - \delta) \log2 -\eta(y) \\
			&=  -(\tau - \delta) \log2 -\left(        \frac{c_0^2}{2c_1}{y}^2 + c_0 \lambda_0^*(y)     \right), \label{eq:internalcrude}
		\end{align}
		and the \textit{optimized internal angle exponent} as
		\begin{align}
			\bar{\psi}_{int}(\beta,\tau) &\triangleq \max_y  {\psi}_{int}(\beta,\tau,y) \\
			&=  -(\tau  - \delta) \log2 -\eta(y^*) \\
			&= -(\tau - \delta) \log2 -\left(        \frac{c_0^2}{2c_1}{y^*}^2 + c_0 \lambda_0^*(y^*)     \right), \label{eq:intexponent}
		\end{align}
		where $y^*$ is determined through $s^*$ from equation (\ref{eq:yands}) and $s^*$ is determined by solving the equation (\ref{eq:findsstar}). From inequality (\ref{eq:internalld})
		 and Lemma \ref{lem:internalexpression} we 
		see that the function $\bar{\psi}_{int}(\beta,\tau)$ satisfies the conditions described in the first part
		of Theorem \ref{thm:firstfaceexponents}.\\

		\end{subsubsection}
		
		\begin{subsubsection}{\textbf{External Angle Exponent}}
			In this subsection we find the optimized external angle exponent $\bar{\psi}_{ext}(\tau)$ that satisfies the statement of Theorem \ref{thm:firstfaceexponents}.
			We proceed similar to the previous section and begin by stating the following lemma from \cite{WeiyuThesis} which provides the expression for the external angle of a face $G$ with respect to the polytope
			 $\mathcal{P}$ in
			 terms of the weights associated with its vertices. We then find an upper bound on the asymptotic exponent of the external angle.
			\begin{lem} \cite{WeiyuThesis} \label{lem:externalexpression}
				The external angle $\gamma(G_0^{l},\mathcal{P})$ is given by
				\begin{align*}
					\gamma(G_0^l,\mathcal{P})  &= \pi^{-\frac{n-l+1}{2}}2^{n-l} \times \\
					& \int_{0}^{\infty}e^{-x^2} \prod_{p=l+1}^{n}\left( \int_{0}^{\frac{w_p x}{\sqrt{\sigma_{1,l}}}} e^{-y_p^2}dy_p\right) dx. \\
				\end{align*}
			\end{lem}
			
			To simplify the expression for the external angle, we define the standard error function as
			\begin{align*}
				\erf(x) = \frac{2}{\sqrt{\pi}} \int_{0}^{x}e^{-t^2}dt
			\end{align*}
			and rewrite the external angle as
			\begin{align*}
				\gamma(G_0^l,\mathcal{P}) =\sqrt{ \frac{\sigma_{1,l}}{\pi}} \int_{0}^{\infty} e^{-\sigma_{1,l}x^2}\prod_{i=l+1}^{n} \erf(w_ix)dx.
			\end{align*}
			Similar to the method used in the previous section
			\begin{align*}
				\sigma_{1,l} &= \sum_{i=1}^{l}w_i^2 = n \left( \int_{0}^{\tau} f^2(x)dx + o(1) \right) \\
				&= nc_2(\tau) + o(1).
			\end{align*}
			where $c_2(\tau)$ is defined as $c_2(\tau) \triangleq \int_{0}^{\tau} f^2(x)dx $.
			Substituting this we have
			\begin{align*}
				&\gamma(G_0^l,\mathcal{P})  = \sqrt{ \frac{\sigma_{1,l}}{\pi}} \times\\
				&\int_{0}^{\infty} \exp \left[  -n \left(    c_2x^2 -  \sum_{i=l+1}^{n}\frac{\log(\erf(w_ix))}{n}  \right)     \right]dx \\
				=& \sqrt{ \frac{\sigma_{1,l}}{\pi}}  \int_{0}^{\infty} \exp[-n\zeta(x)]dx.
			\end{align*}
			where $\zeta(x)$ is defined as $\zeta (x) \triangleq \left(    c_2x^2 - \frac{1}{n} \sum_{i=l+1}^{n}\log(\erf(w_ix))  \right)$.
			Again using Laplace's method we get 
			\begin{align} \label{externalcramer}
				\gamma(G_0^l,\mathcal{P}) \leq R_n \exp[-n \zeta(x^*)],
			\end{align}
			where $x^*$ is the minimizer of $\zeta(x)$ and $n^{-1}\log(R_n) = o(1)$.
			The minimizing $x^*$ satisfies $2c_2x^* =  G_0^{'}(x^*),$ where
			\begin{align*}
				G_0(x) \triangleq \frac{1}{n} \sum_{i=l+1}^{n}\log(\erf(w_ix)).
			\end{align*}
			We first approximate $G_0(x)$ as follows:
			\begin{align*}
				G_0(x) &= \frac{1}{n} \sum_{i=l+1}^{n}\log(\erf(w_ix)) \\
				&= \frac{1}{n} \sum_{i=l+1}^{n}\log(\erf(f(i/n)x)) \\
				&= \int_{\tau}^{1}\log( \erf( xf(y))  )dy +o(1).
			\end{align*}
			So the minimizing $x^*$ can be computed up to an error of $o(1)$ by solving the equation
			\begin{align}
				2c_2x^* = \int_{\tau}^{1} \frac{f(y)\erf^{'}(x(f(y))}{\erf(xf(y))}dy \label{eq:xstar}.
			\end{align}
			We define the \textit{external angle exponent} as
			\begin{align} \label{eq:externalcrude}
				\bar{\psi}_{ext}(\tau,x) \triangleq -\zeta(x)
			\end{align}
			and the \textit{optimized external angle exponent} as 
			\begin{align}
				&\bar{\psi}_{ext}(\tau) \triangleq \max_x \bar{\psi}_{ext}(\tau) = -\zeta(x^*)  \\
				&= -\left(    c_2{x^*}^2 -   \int_{\tau}^{1}\log( \erf( x^*f(y))  )dy\right) \label{eq:extexponent}.
			\end{align}
			where $x^*$ can be obtained by solving equation (\ref{eq:xstar}).

		\end{subsubsection}
		From (\ref{externalcramer}) it is clear that this function satisfies the conditions in the second part of Theorem \ref{thm:firstfaceexponents}.
		This completes the proof of both parts of the theorem.
		
	\end{subsection}
	
	\begin{subsection}{Recovery Threshold for the family of leading faces} \label{subsec:firstfamilythreshold}
		In this section we use the bounds on the asymptotic exponents of the internal and external angle from Theorem \ref{thm:firstfaceexponents} to find an upper bound for $\mathbf{P}(E|\mathbf{x} \in F)$ for
		$F \in \mathcal{F}_1$. The main result of this section is the following theorem.
		\begin{thm} \label{thm:firstfacetotalexponent}
			 Let $F = F_0^k$ for $k = \delta n$.
			 $F_{\delta} \in \mathcal{F}_1$ be the face with $k = \delta n$. There exists a function $\bar{\psi}_{tot}(\delta)$ which we call the total exponent of $F_0^k$
			 such that given $\epsilon > 0$, there exists $n(\epsilon)$ such that for all $n > n(\epsilon)$, 
			 $\frac{1}{n} \log (\mathbf{P}(E|\mathbf{x} \in F_0^k)) < -n \bar{\psi}_{tot}(\delta) + \epsilon$.
			 The function $\bar{\psi}_{tot}(\delta)$ can be computed explicitly as described in (\ref{eq:totalimprovedexponent}).
		\end{thm}
		Using the decomposition equation (\ref{eq:decomposition}) 
		and the fact that $\beta(F,G)$ is non-zero only if $F \subseteq G$, we get
		\begin{align*}
			\mathbf{P}(E|\mathbf{x} \in F) \leq 2 \sum_{(l > m)} \sum_{(G \supseteq F, \ G \in J(l)) }  \beta(F,G) \gamma(G,\mathcal{P}).
		\end{align*}
		Recall that $J(r)$ is the set of all $r$-dimensional faces of $\mathcal{P}$.
		
		To proceed, we will need the following useful lemma (proof can be found in the appendix).
		\begin{lem} \label{lem:basic_exponent_monotonicity}
			Among all $(l-1)$-dimensional faces $G$ of $\mathcal{P}$ satisfying $F_0^k \subseteq G$, the face that maximizes $\beta(F_0^k,G)$ and $\gamma(G,\mathcal{P})$ is the 
			one with $\frac{1}{w_1} e_1, \frac{1}{w_2} e_2, \ldots , \frac{1}{w_l}e_l$
			as its vertices.
		\end{lem}
		 Using Lemma \ref{lem:basic_exponent_monotonicity}, we get $\mathbf{P}(E|\mathbf{x} \in F_0^k)  $	 
		\begin{align*}
			 \leq & \sum_{l=m+1}^n {n-k \choose l-k} 2^{l-k}  \beta(F_0^k,G_0^l) \gamma(G_0^l,\mathcal{P}) \\
			 \leq &(n-m) \max_{\{l : l >m\}}      {n-k \choose l-k}  2^{l-k}  \beta(F_0^k,G_0^l) \gamma(G_0^l,\mathcal{P}).
		\end{align*}
		Using $\tau = \frac{l}{n}$, $\delta = \frac{k}{n}$ and Stirling's approximation for factorials, it can be shown that
		\begin{align*}
			\frac{1}{n}\log {n-k \choose l-k    }      \to  (1 - \delta) H \left(   \frac{\tau-\delta}{1-\delta}   \right),
		\end{align*}
		where $H(x)$ is the entropy function with base $e$ defined by $H(x) = -x \log(x) - (1-x) \log(1-x)$. Define the \textit{combinatorial exponent} $\psi_{com}(\delta,\tau)$ as
		\begin{align*}
			\psi_{com}(\delta,\tau) \triangleq (1 - \delta) H \left(   \frac{\tau-\delta}{1-\delta}   \right) + (\tau - \delta) \log 2.
		\end{align*}
		So, we conclude,
		\begin{align*}
			\frac{1}{n} \log(\mathbf{P}(E| \mathbf{x} \in F)) &\leq \max_{\{\gamma : \gamma \geq \alpha\}} \   \bar{\psi}_{com}(\delta,\tau) \\
			&+ \bar{\psi}_{int}(\delta,\tau) + \bar{\psi}_{ext}(\tau) + o(1).
		\end{align*}
		The function $\bar{\psi}_{tot}(\delta)$ defined by 
		\begin{align*}
			\bar{\psi}_{tot}(\delta) \triangleq  \max_{\{\tau : \gamma \geq \alpha\}} \   &\bar{\psi}_{com}(\delta,\tau) \\
			&+ \bar{\psi}_{int}(\delta,\tau) + \bar{\psi}_{ext}(\tau)
		\end{align*}
		satisfies the condition given in Theorem \ref{thm:firstfacetotalexponent}, and hence establishes the theorem.
		\end{subsection}
		
		\begin{subsection}{Obtaining tighter exponents}
		In Section \ref{subsec:firstfamilythreshold}, we developed a method that provides an upper bound on $\mathbf{P}(E|\mathbf{x} \in F_0^k)$ in terms of internal and external angle exponents.
		For a given value of $l$, we bounded $\beta(F_0^k,G)$ and $\gamma(G,\mathcal{P})$ for each $G \in J(l)$ by 
		$\beta(F_0^k,G_0^l) \gamma(G_0^l,\mathcal{P})$ by using Lemma \ref{lem:basic_exponent_monotonicity}.
		This bounding step is generally rather overly conservative because
		it does not take into account the variation of the function $f(.)$ over the complete interval $[0,1]$. The quality of bound is especially poor for choices of function $f(.)$ which are rapidly increasing. 
		To improve on the bound, we can use a simple technique which allows us to take into account the variation of $f(u)$ with respect to $u$ more accurately over the whole interval.
		
		Divide the set of indices $k+1, \ldots n$ into two parts with $T_1 \triangleq \{k+1, \ldots ,\frac{n+k}{2} \}$ and $T_2 \triangleq \{  \frac{n+k}{2}+1 , \ldots , n  \}$.
		For a particular $l$, let $G$ have $l_1$ vertices in $T_1$ and $l_2$ vertices in $T_2$. Using Lemma \ref{lem:basic_exponent_monotonicity},  among all faces $G$ with the values of $l_1$
		and $l_2$ specified as above, the choice that maximizes $\beta(F_0^k,G)$ and $\gamma(G,\mathcal{P})$ is the face $G$ with vertices given by the indices $1, \ldots ,l_1, \frac{n+k}{2} +1
		, \ldots , \frac{n+k}{2} + l_2$. Using this we get a revised upper bound on $\mathbf{P}(E|\mathbf{x} \in F_0^k)$
		\begin{align*}
			&\leq  \sum_{l=m+1}^{n} \sum_{l_1+ l_2 = l} {{ \frac{n-k}{2}    } \choose l_1     } {{ \frac{n-k}{2}    } \choose l_2     }  \beta(F_0^k,G) \gamma(G,\mathcal{P}) \\
			&\leq  (n-m)n \max_{l_1+ l_2 \geq m - k} {{ \frac{n-k}{2}    } \choose l_1     } {{ \frac{n-k}{2}    } \choose l_2     } \times \\
			& \hspace{1.2in} \beta(F_0^k,G_2) \gamma(G_2,\mathcal{P}), 
		\end{align*}
		Define $\gamma_1 = \frac{l_1}{n}$ and $\gamma_2 = \frac{l_2}{n}$. Then the revised bound on the exponent of $\mathbf{P}(E|\mathbf{x} \in F_0^k)$ can be obtained as 
		\begin{align*}
			\frac{1}{n}&\log(\mathbf{P}(E|\mathbf{x} \in F)) \\
			&\leq  \max_{\gamma_1 
			+ \gamma_2 \geq \alpha - \delta} \frac{1-\delta}{2}\left[H\left( \frac{2\gamma_1}{1-\delta} \right) + H\left( \frac{2\gamma_2}{1-\delta} \right)\right] \\
			&+  \psi_{int}(\delta,	\gamma_1,\gamma_2) + \psi_{ext}(\gamma_1,\gamma_2) + o(1) \\
			&= \max_{\gamma_1 
			+ \gamma_2 \geq \alpha - \delta} \psi_{com}(\delta,\gamma_1,\gamma_2) +  \psi_{int}(\delta,	\gamma_1,\gamma_2) \\
			& \hspace{0.8in}+ \psi_{ext}(\gamma_1,\gamma_2) + o(1).
		\end{align*}
		where we define $ \psi_{com}(\delta,\gamma_1,\gamma_2) \triangleq \frac{1-\delta}{2}H\left( \frac{2\gamma_1}{1-\delta} \right) + \frac{1-\delta}{2}H\left( \frac{2\gamma_2}{1-\delta} \right)$ and call it the 
		\textit{combinatorial exponent} 
		The expressions for $\psi_{int}(\delta,	\gamma_1,\gamma_2)$ and $\psi_{ext}(\gamma_1,\gamma_2)$ are can be obtained by using the methods of the previous subsection
		 (we will give the exact expressions shortly). The function
		 \begin{align*}
		 	\bar{\psi}_{tot}(\delta) &\triangleq \max_{\gamma_1 
			+ \gamma_2 \geq \alpha - \delta} \psi_{com}(\delta,\gamma_1,\gamma_2) \\
			& \hspace{0.5in}+  \psi_{int}(\delta,	\gamma_1,\gamma_2) + \psi_{ext}(\gamma_1,\gamma_2).
		 \end{align*}
		now satisfies the role of the \textit{total exponent} in Theorem \ref{thm:firstfacetotalexponent}.\\

		We can repeat the above argument for any $r$ by dividing the indices denoted by $k+1 , \ldots , n$ into $r$ parts. Define the quantities $\gamma_i, \ i = 1, \ldots r$ as
		$\gamma_i \triangleq \frac{l_i}{n}$ where $l_i$ is the number of indices of face $G$ in the $i^{th}$ interval. Also define
		 $h_i \triangleq r \gamma_i$ and let $\mathbf{h} = (h_1, \ldots ,h_r)^t$. 
		The total exponent is now given by
	\begin{align*}
			\bar{\psi}_{tot} = \max_{\mathbf{h}} \quad & \bar{\psi}_{com}(\mathbf{h}) + \bar{\psi}_{int}(\mathbf{h}) + \bar{\psi}_{ext}(\mathbf{h}) \\
		\mbox{subject to} \quad &\frac{1}{r} \sum_{i=1}^{r} h_i \geq \alpha - \delta,
	\end{align*}
	 where
	 \begin{align*}
	 	\bar{\psi}_{int}(\mathbf{h}) = \max_y \mbox{ } \psi_{int}(\mathbf{h},y), \quad \mbox{as in (\ref{eq:internalcrude}) and}
	 \end{align*}
	 \begin{align*}
	 	\bar{\psi}_{ext}(\mathbf{h}) = \max_x \mbox{ } \psi_{ext}(\mathbf{h},x), \quad \mbox{as in (\ref{eq:externalcrude}) }.
	 \end{align*}
	 The dependence of the exponent functions on other variables has been suppressed for compactness.

	 So, the total exponent can be obtained by the 
	 following maximization:
	 \begin{align}
	 	\bar{\psi}_{tot} = \max_{\mathbf{h},x,y} \quad &  \psi_{com}(\mathbf{h}) + \psi_{int}(\mathbf{h},y) + \psi_{ext}(\mathbf{h},x)   \label{eq:totalimprovedexponent} \\
		\mbox{subject to} \quad & \frac{1}{r} \sum_{i=1}^{r} h_i \geq \alpha - \delta. \notag
	 \end{align}
	 
	 	 We now compute the expressions for each of the functions appearing in the above maximization.
	For the subsequent derivation, we define $f_i \triangleq f \left(   \delta + \frac{(i-1)(1-\delta)}{r}  \right)$. \\
	
         \begin{subsubsection}{\textbf{Combinatorial Exponent}} \label{sec:firstfacecombinatorialexponent}
		The combinatorial exponent is given by
		\begin{align*}
			\psi_{com} &= \frac{1}{r} \sum_{i=1}^{r} (1-\delta) H \left(   \frac{r \gamma_i}{1-\delta} \right) + \left( \sum_{i=1}^{r} \gamma_i - \delta \right) \log 2\\
			&= \frac{1-\delta}{r}  \sum_{i=1}^{r} H \left(   \frac{ h_i}{1-\delta} \right) + \left( \sum_{i=1}^{r} r h_i - \delta \right) \log 2.
		\end{align*}
	\end{subsubsection}
	
	For the internal and external angle exponents, in addition to obtaining their expressions, we will also bound them suitably by analytically simpler expressions so that the
	optimization problem described in (\ref{eq:totalimprovedexponent}) becomes more tractable.\\
	
	\begin{subsubsection}{\textbf{External Angle Exponent}}  \label{sec:firstfaceexternalangle}
	We start with 
	\begin{align*}
		\zeta(x) = c_2 x^{2} -G_0(x),
	\end{align*}
		where
		\begin{align*}
			c_2 &= \int_{0}^{\delta} f^2(u) du + \sum_{i=1}^{r} \int_{\delta + \frac{(i-1)  (1-\delta)}{r}}^{\delta + \frac{(i-1)  (1-\delta)}{r} + \frac{h_i}{r}} f^2(u) du \\
			G_0(x)& = \int_{\delta}^{1} \log(\erf(xf(u))) du \\
			&- \sum_{i=1}^{r} \int_{\delta + \frac{(i-1)  (1-\delta)}{r}}^{\delta + \frac{(i-1)  (1-\delta)}{r} + \frac{h_i}{r}} \log(\erf(xf(u))) du.
		\end{align*}
		As $f(.)$ is an increasing function, the integral appearing above in the expression of $c_2$ can be bound by its left Riemann sum.
		\begin{align*}
			c_2 \geq 	\int_{0}^{\delta} f^2(u) du + \frac{1}{r} \sum_{i=1}^{r}  f_{i}^{2} {h_{i}} \triangleq \bar{c_2}. 
		\end{align*}
		Similary,
		\begin{align*}
			G_0(x) &\leq \int_{\delta}^{1} \log(\erf(xf(u))) du \\
			&-  \frac{1}{r} \sum_{i=1}^{r} \log(\erf(x f_i)) \triangleq \bar{G_0}(x).
		\end{align*}
		So,
		\begin{align*}
			\zeta(x) \geq \bar{c_2}x^2 - \bar{G_0}(x) \triangleq \bar{\zeta(x)} .
		\end{align*}
		Combining we obtain a simplified expression for the external angle exponent as		
		\begin{align*}
			 \psi_{ext}(x) = -\bar{\zeta}(x) = -\left( \bar{c}_2 x^{2} - \bar{G}_0(x) \right).
		\end{align*}	
		The optimized external angle exponent is then given by 
		\begin{align*}
			\bar{\psi}_{ext} = \max_{x} \psi_{ext}(x) .
		\end{align*}	
		\\
		
	\end{subsubsection}
	
	\begin{subsubsection}{\textbf{Internal Angle Exponent}}  \label{sec:firstfaceinternalangle}
	We start with 
	\begin{align*}
		\eta(y) = \frac{c_0^2}{2c_1}y^2 + c_0 \lambda_0^{*}(y)
	\end{align*}
	where
		\begin{align*}
			c_0 &=  \sum_{i=1}^{r} \int_{\delta + \frac{(i-1)  (1-\delta)}{r}}^{\delta + \frac{(i-1)  (1-\delta)}{r} + \frac{h_i}{r}} f(u) du \\
			&\geq  \frac{1}{r} \sum_{i=1}^{r}  f_{i} {h_{i}} \triangleq \bar{c_0}.
		\end{align*}
		$\lambda_0^{*}(y)$ is the convex conjugate of $\lambda_0(.)$ and is given by,
		\begin{align*}
			\lambda_0^{*}(y) &= \max_s  \mbox{  }  sy - \lambda_0(s).
		\end{align*}
		We are interested in only in the region $s \leq 0$. In this region we have,
		\begin{align*}
			\lambda_0(s) &=  \frac{1}{c_0} \sum_{i=1}^{r} \int_{\delta + \frac{(i-1)  (1-\delta)}{r}}^{\delta + \frac{(i-1)  (1-\delta)}{r} + \frac{h_i}{r}} \lambda(sf(u)) du \\
			&\leq  \frac{1}{r \bar{c}_0}
			 \sum_{i=1}^{r} h_i \lambda(sf_i) \triangleq \bar{\lambda}_0(s),
		\end{align*}
		which follows from the fact that $s \leq 0$ and $\lambda(u)$ is an increasing function of $u$.
		This gives
		\begin{align*}
			\max_s \mbox{  } sy - \lambda_0(s) \geq \max_s \mbox{  } sy - \bar{\lambda}_0(s),
		\end{align*}
		and hence
		\begin{align*}
			\lambda_0^{*}(y) \geq \bar{\lambda}_0^{*}(y) = \max_s \mbox{  } sy - \bar{\lambda}_0(s).
		\end{align*}
		So, we conclude
		\begin{align*}
			\eta(y) \geq  \frac{\bar{c}_0^2}{2c_1}y^2 + \bar{c}_0 \bar{\lambda}_0^{*}(y) \triangleq \bar{\eta}(y).
		\end{align*}
		From the above we obtain a simplified expression for the internal angle exponent as 
		\begin{align*}
			\psi_{int}(y) &=  -(r \sum_{i=1}^{r}h_i - \delta) \log 2 -\bar{\eta}(y) \\
			&= -(r \sum_{i=1}^{r}h_i  - \delta) \log 2 - \left(  \frac{\bar{c}_0^2}{2c_1}y^2 + \bar{c}_0 \bar{\lambda}_0^{*}(y)  \right).
		\end{align*}
		The optimized internal angle exponent is then obtained as 
		\begin{align*}
			\bar{\psi}_{int} = \max_{y} \psi_{int}(y).
		\end{align*}
		\end{subsubsection}
	
	To show that the maximization with respect to $\mathbf{h}$ in (\ref{eq:totalimprovedexponent}) can be performed efficiently, we will show that $\psi_{tot}(\mathbf{h},x,y)$ is
	 a concave function of $\mathbf{h}$ for fixed value of $x$ and $y$. 
	 This follows from the following lemma. The proof can be found in the appendix.
	\begin{lem} \label{lem:internalexponentcurvature}
		The combinatorial, internal and external angle exponents are concave functions of $\mathbf{h}$ for fixed values of $x$ and $y$.
	\end{lem}
	
	The parameter $r$ which governs the number of evaluation points we use to compute the angle exponents gives us a method to control the accuracy of the computed exponents. If we can obtain a value
	 $\bar{\delta}$ such that for all $\delta \leq \bar{\delta}$ we have $\bar{\Psi}_{tot}(\delta) < 0$, then Theorem \ref{thm:firstfacetotalexponent} guarantees that whenever $\delta \leq \bar{\delta}$, weighted 
	 $\ell_1$-minimization recovers the corresponding sparse signal with overwhelming probability. We call this $\bar{\delta}$ the guaranteed bound on recoverable sparsity levels $\delta$.
	 Increasing $r$ results in a tighter guaranteed bound but
	at the expense of an increased cost in performing the optimization in (\ref{eq:totalimprovedexponent}).
	In Figure \ref{fig:thresholdvsm} we
	 show via simulation how this bound improves with the parameter $r$. For this, we fix a compression ratio $\alpha = \frac{m}{n} = 0.5$. The weights $w_i$ are chosen as $w_i = 1 + \frac{i}{n}$ which corresponds to
	 the weight function $f(u) = 1 + u$. 
	Note how the bound improves with increasing value of $r$ but saturates quickly. This indicates that a moderately large value of $r$ (e.g. $r=30$ from the figure) will suffice for accurate enough angle exponent
	computations. 
		
	\begin{figure}
		\begin{center}
		\includegraphics[scale = 0.45]{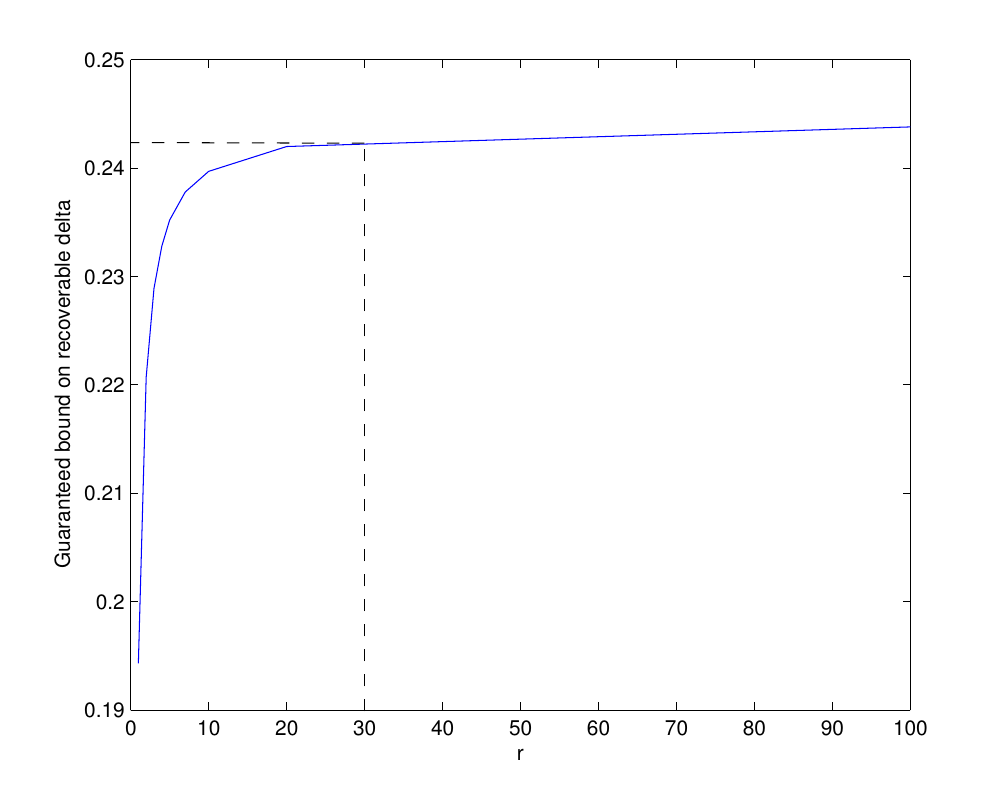}
		\end{center}
		\caption{Guaranteed bound on recoverable $\delta$ vs $r$ for compression ratio $\alpha = 0.5$ and weight function $f(u) = 1 + u$
		 computed using the methods of Section \ref{sec:analysisofl1}. As $r$ increases the computed bound also increases indicating an improvement in the tightness of the bound.
		  Also since the improvement saturates fairly fast, we can use 
		a moderately large value of $r$ to obtain accurate enough angle exponents. The value $r=30$ in the figure represents such a choice.}
		\label{fig:thresholdvsm}
	\end{figure}

\end{subsection}

\end{section}

\begin{section}{Proof of Theorem \ref{thm:mainresult}} \label{sec:withprob}
		In this section, we establish the main result of this paper which is Theorem \ref{thm:mainresult}. We do this in two steps. We first characterize the properties of a typical signal drawn from our model and provide a 
		large deviation result on the probability that a signal drawn from our model does not satisfy this property.  For signals which have this typical property, we
		 then use the analysis method of Section \ref{sec:analysisofl1} to find the exponent
		$\psi_{tot}(p,f)$ in the statement of Theorem \ref{eq:weightedl1}.		
		We start by writing
		\begin{align*}
			\mathbf{P}(E) = \sum_{F \in \mathcal{F}} \mathbf{P}({\mathbf{x} \in {F}}) \mathbf{P}(E | \mathbf{x} \in F).
		\end{align*}
		To analyze this expression we will further split the sum into faces belonging to different ``classes" which we describe below.
		Divide the set of indices from $1, \ldots, n$ into $r$ equal parts with $I_i$ denoting the $i^{th}$ interval of indices.
		Any face $F$ of the skewed cross polytope $\mathcal{P}$ can be fully specified by the index of its vertices (up to the signs 
		of the vertices).  For a given face $F$, let the set of indices representing the face $F$ be denoted by $I(F)$.
		For a given $\mathbf{k} = (k_1, \ldots, k_r)$ denote by $\mathcal{F}_r(\mathbf{k})$, the set of all faces of $\mathcal{P}$ with $|I(F) \cap I_i| = k_i$. Also let $g_i = k_i \frac{r}{n}$.
		Recall that we denote by $E$ the failure event, i.e. the event that the weighted $\ell_1$-minimization does not produce the correct solution. Then we have
		\begin{align}
			\mathbf{P}(E) = \sum_{\mathbf{k}} \sum_{F \in \mathcal{F}_r(\mathbf{k}) } \mathbf{P}(\mathbf{x} \in F) \mathbf{P}(E | \mathbf{x} \in F). \label{eq:mainerror}
		\end{align} 
		We can just consider one representative among the $2^{\sum_i k_i}$ faces created by the different sign patterns. This is because, by the symmetry of the problem,
		 all the faces have the same probabilities $\mathbf{P}(\mathbf{x} \in F)$ and $\mathbf{P}(E|\mathbf{x} \in F)$. For simplicity we always choose this representative as the face that lies in the first orthant.
		\begin{align*}
			\mathbf{P}(\mathbf{x} \in F) = & \left( \prod_{i=1}^{r} \prod_{j \in I(F) \cap I_i }  \frac {p(j/n)}{1 - p(j/n)  }  \right ) \\
			&  \times \prod_{i=1}^{n} (1 - p(j/n)).
		\end{align*} 
		The function $\frac{x}{1-x}$ is an increasing function of $x$ for $x \in (0,1)$. So,
		\begin{align*}
			\mathbf{P}(\mathbf{x} \in F) \leq  \left(   \prod_{i=1}^{r} {\left( \frac{p\left(  \frac{i-1}{r} \right)}{1-p\left(  \frac{i-1}{r} \right)} \right)}^{k_i}     \right )  \prod_{i=1}^{n} (1 - p(j/n)).
		\end{align*} 
		Denote the right hand side of the above inequality be $\mathbf{P}(\mathbf{k})$, which means for $\mathbf{x} \in \mathcal{F}_r(\mathbf{k})$ we have $\mathbf{P}(\mathbf{x} \in F) \leq \mathbf{P}(\mathbf{k})$.
		Define
		\begin{align*}
			\mathbf{P}(\mathbf{x} \in \mathcal{F}_r(\mathbf{k})) \triangleq \mathbf{P}(\mathbf{x} \in F \, \mbox{for some } F \in \mathcal{F}_r(\mathbf{k}) ).
		\end{align*}
		Then
		\begin{align*}
			\mathbf{P}(\mathbf{x} \in \mathcal{F}_r(\mathbf{k})) &\leq \sum_{F \in  \mathcal{F}_r(\mathbf{k}) } \mathbf{P}(\mathbf{x} \in F) \\
			& \leq 	\left(  \prod_{i=1}^{r}  { {\frac{n}{r}} \choose {\frac{g_i}{r}n}} \right) P(\mathbf{k}).
		\end{align*}
		By Lemma \ref{lem:basic_exponent_monotonicity}, among all faces $F \in \mathcal{F}_r(\mathbf{k}) $, the one
		 which maximizes $\mathbf{P}(E|\mathbf{x} \in F)$ is the one obtained by stacking all the indices to the right of each interval.
		We denote this maximum probability by $P(E|\mathbf{k})$.
		Combining the above we get
		\begin{align}
			\mathbf{P}(E) \leq \sum_{\mathbf{k}}   \mathbf{P}(\mathbf{x} \in \mathcal{F}_r(\mathbf{k})) \mathbf{P}(E|\mathbf{k}). \label{eq:errorsum}
		\end{align}
		As the function $p(u)$ is monotonically decreasing, $ \mathbf{P}(\mathbf{x} \in \mathcal{F}_r(\mathbf{k}))$ is not same for all $\mathbf{k}$.
		We now proceed to show that $ \mathbf{P}(\mathbf{x} \in \mathcal{F}_r(\mathbf{k}))$ is significant only when $\mathbf{x}$ takes values in a ``typical" set. For values of $\mathbf{k}$ outside the typical set
		$ \mathbf{P}(\mathbf{x} \in \mathcal{F}_r(\mathbf{k}))$ is exponentially small in $n$ and can be ignored in the sum in (\ref{eq:errorsum}). For $\mathbf{k}$ in the typical set, we will use the methods of 
		Section \ref{sec:analysisofl1} to bound $\mathbf{P}(E|\mathbf{k})$ and hence $\mathbf{P}(E)$.	
		
		The following Lemma provides bounds for $\mathbf{P}(\mathbf{x} \in \mathcal{F}_r(\mathbf{k}))$ which will motivate the definition of typicality that follows.	
		\begin{lem} \label{typicality1}
			Let  $D( q || p)$ denote the Kullback-Leibler distance between two Bernoulli random variables with probability of success given by $q$ and $p$ respectively. Define 
			$\bar{p}_i = r \int_{\frac{i}{r}}^{\frac{i+1}{r}} p(u) du$. If we have $\frac{1}{r} \sum_{i=1}^{r} D(g_i || \bar{p}_i) > \epsilon$, then there exists $n_0$, such that for all $n > n_0$,
			\begin{align*}
				\mathbf{P}(\mathbf{x} \in \mathcal{F}_r(\mathbf{k}))< e^{-b(\epsilon)n},
			\end{align*}
			where $b(\epsilon)$ is a positive constant.
		\end{lem}
		\begin{proof}
		We have
		\begin{align*}
			\mathbf{P}(\mathbf{x} \in \mathcal{F}_r(\mathbf{k})) = \left(  \prod_{i=1}^{r}  { {\frac{n}{r}} \choose {\frac{g_i}{r}n}} \right) P(\mathbf{k}).
		\end{align*} 	
		Then,
		\begin{align*}
			&\frac{1}{n} \log \left(\mathbf{P}(\mathbf{x} \in \mathcal{F}_r(\mathbf{k})) \right) = \frac{1}{n} \sum_{i=1}^{r} {\frac{n}{r} \choose \frac{g_i n}{r}  } + \frac{1}{n} \log (P(\mathbf{k})) \\
			=&  \frac{1}{n} \sum_{i=1}^{r} {\frac{n}{r} \choose \frac{g_i n}{r}  }  + \frac{1}{n} \sum_{i=1}^{r} \frac{g_i n}{r} \log\left(   \frac{p_i}{1-p_i}  \right) \\
			+& \frac{1}{n} \sum_{j=1}^{n} \log \left( 1 - p\left( \frac{j}{n}\right)\right).
		\end{align*} 
		Letting $n \rightarrow \infty$,
		\begin{align*}
			\lim_{n \rightarrow \infty}& \frac{1}{n} \log(\mathbf{P}(\mathbf{x} \in \mathcal{F}_r(\mathbf{k}))) \\
			&=  \frac{1}{r} \sum_{i=1}^{r}H(g_i) + \frac{1}{r} \sum_{i=1}^{r}g_i \log \left(   \frac{p_i}{1-p_i}  \right) \\
			& + \int_{0}^{1} \log(1 - p(u))du \\
			&=\frac{1}{r} \sum_{i=1}^{r} H(g_i) + \left( g_i \log(p_i) + (1-g_i) \log(1- p_i)\right)\\
			& -\log(1 - p\left( \frac{i-1}{r}\right))  + \int_{0}^{1} \log(1 - p(u))du \\
			&= \frac{1}{r} \sum_{i=1}^{r} -D(q_i || p_i) + \Delta(r),
		\end{align*}
		where $\Delta(r)$ is the error in the Riemann sum given by  $\Delta(r) = \left[ \int_{0}^{1} \log(1 - p(u))du - \frac{1}{r} \sum_{i=1}^{r} \log(1 - p\left( \frac{i-1}{r} \right))  \right]$. 
		Now further divide each of the $r$ intervals into $t$ equal parts thus forming a total of $tr$ intervals. Let $g_{t,i} \frac{n}{r}$ denote the number of indices of face $F$ in the 
		$(tr)^{th}$ interval counted according to the new partition. Summing up the number of indices in all the intervals of length $\frac{1}{tr}$ contained in an interval of length
		 $\frac{1}{r}$ gives $g_i = \frac{1}{k} \sum_{j=ti+1}^{t(i+1)} g_{t,j}$. Also let $p_{r,j} = p \left( \frac{i}{r} \right)$.
		 The calculations just carried out above can be repeated for this
		new partition to give 
		\begin{align*}
			\lim_{n \rightarrow \infty}& \frac{1}{n} \log (\mathbf{P}(\mathbf{x} \in \mathcal{F}_r(\mathbf{k})))\\
			& \leq  -\frac{1}{tr} \sum_{i=1}^{r} \sum_{j=1}^{t}  D(g_{t,ti + j} || p_{t,ti+j}) + \Delta(tr).
		\end{align*}
		We now bound the term $ \sum_{j=1}^{t}  D(g_{t,ti + j} || p_{t,ti+j})$.
			\begin{align*}
				\sum_{j=1}^{t} &  D(g_{t,ti + j} || p_{t,ti+j}) \\
				&= \sum_{j=1}^{t} g_{t,ti + j} \log \left( \frac{g_{t,ti + j}}{p_{t,ti+j}}\right) \\
				&+   \sum_{i=1}^{r} (1-g_{t,ti + j}) \log \left(\frac{1-g_{t,ti + j}}{1-p_{t,ti+j}} \right).
			\end{align*}
			Using the log-sum inequality, we obtain,
			\begin{align*}
				\frac{1}{t} &\sum_{j=1}^{t} D(g_{t,ti + j} || p_{t,ti+j}) \\
				&\geq  \left(  \frac{\sum_{j=1}^{t}g_{t,ti + j}}{t}  \right) \log \left(  \frac{\sum_{j=1}^{t} g_{t,ti + j}}{\sum_{j=1}^{t} p_{t,ti+j}} \right) \\
				&+ \left(  \frac{\sum_{j=1}^{t}1- g_{t,ti + j}}{t}  \right) \log \left(  \frac{\sum_{j=1}^{t}1- g_{t,ti + j}}{\sum_{i=1}^{r}1- p_{t,ti+j}} \right) \\
				&= g_i \log \left(  \frac{g_i}{\bar{p}_{i,t} } \right) + (1-g_i) \log \left(  \frac{1-g_i}{1- \bar{p}_{i,t}  } \right) \\
				&= D ( g_i || \bar{p}_{i,t} ),
			\end{align*}
			where, $\bar{p}_{i,t} = \frac{1}{t} \sum_{j=1}^{t} p_{t,ti+j} $.
			Using this,
			\begin{align*}
				\lim_{n \rightarrow \infty} & \frac{1}{n} \log (\mathbf{P}(\mathbf{x} \in \mathcal{F}_r(\mathbf{k}))) \\
				&\leq  -\frac{1}{r} \sum_{i=1}^{r}D ( g_i ||\bar{p}_{i,t} ) + \Delta(tr).
			\end{align*}
			Since this is true for every $t$, we let $t \rightarrow \infty$ to get
			\begin{align*}
				\lim_{n \rightarrow \infty} \frac{1}{n} \log (\mathbf{P}(\mathbf{x} \in \mathcal{F}_r(\mathbf{k}))) \leq -\frac{1}{r} \sum_{i=1}^{r}D ( g_i ||\bar{p}_i ).
			\end{align*}
			So, if the condition of the lemma is satisfied, then
			\begin{align*}
				\lim_{n \rightarrow \infty} \frac{1}{n} \log (\mathbf{P}(\mathbf{x} \in \mathcal{F}_r(\mathbf{k}))) \leq - \frac{\epsilon}{r}.
			\end{align*}
			and the claim in the lemma then follows.
	\end{proof}

	Lemma \ref{typicality1} motivates the following natural definition of $\epsilon$-$typicality$.
	\begin{defin}
		Given $\epsilon > 0$, define $\mathbf{k} = (k_1, \ldots, k_r)$ to be $\epsilon$-typical if $\frac{1}{r} \sum_{i=1}^{r}D ( g_i ||\bar{p}_i ) \leq \epsilon$.
	\end{defin}

		Using this definition of $\epsilon$-$typicality$, we can now bound the probability of failure $\mathbf{P}(E)$ by using (\ref{eq:errorsum}) as
		\begin{align*}
			\mathbf{P}(E) \leq& \underbrace{\sum_{\mathbf{k} \, is \, \epsilon \,typical} \mathbf{P}(\mathbf{x} \in \mathcal{F}_r(\mathbf{k})) \mathbf{P}(E|\mathbf{k})}_{I} \\
			&+ 
			\underbrace{\sum_{\mathbf{k} \, not \, \epsilon\,typical} \mathbf{P}(\mathbf{x} \in \mathcal{F}_r(\mathbf{k})) \mathbf{P}(E|\mathbf{k})}_{II}. 
		\end{align*}
		For a fixed value of $r$, the number of possible values of $k_1, k_2, \ldots, k_r$ is bounded by $n(n+r-1)^{r-1}$. Since $\mathbf{P}(E|\mathbf{k}) \leq 1$, the second term $II$ can be bounded as
		\begin{align*}
			II \leq n(n+r-1)^{r-1} e^{-c(\epsilon)n},
		\end{align*}
		for some positive constant $c(\epsilon)$.
		As $\lim_{n \rightarrow \infty} \frac{1}{n}  \log \left(n(n+r-1)^{r-1} \right)  = 0$, there exists $n_0$, such that for all $n > n_0$, 
		\begin{align*}
			II \leq e^{-c_0(\epsilon)n},
		\end{align*}
		for some positive constant $c_0(\epsilon)$. This allows us to only consider the first term $I$ for $\mathbf{k}$ which are $\epsilon$-typical. 
		Among all $\mathbf{k}$ $\epsilon$-$typical$, let $F_{\epsilon,p}^*$ be the face which maximizes the probability $\mathbf{P}(E | \mathbf{x} \in F)$.
		Then,
		\begin{align*}
			I \leq \mathbf{P}(E | \mathbf{x} \in F_{\epsilon,p}^*).
		\end{align*}
		This gives
		\begin{align*}
			\lim_{n \rightarrow \infty} \frac{1}{n} \log(I) \leq \lim_{n \rightarrow \infty} \frac{1}{n} \log(\mathbf{P}(E | \mathbf{x} 
			\in F_{\epsilon,p}^*)).
		\end{align*}
		Also, since we can choose $\epsilon$ to be as small as we want,
		we can essentially just consider the case with $\epsilon = 0$. This gives us $\mathbf{k} = (k_1, \ldots ,k_r)$ with $k_i = \bar{p}_i$. This defines a fixed face $F_p$ which we call the typical face of $\mathcal{P}$ and
		 we only need to bound the term $\mathbf{P}(E|\mathbf{x} \in F_p)$ which can be done by a simple generalization of the methods in Section \ref{sec:analysisofl1}.
		 The corresponding expressions for the optimized internal and external angle exponents can be found in the Appendix.
\end{section}

\begin{section}{Choosing the weight function} \label{sec:choosing}
	 In this section, we provide a method for choosing the weight function $f(.)$ used in the weighted $\ell_1$-minimization based on the probability function $p(.)$. 	
	For a given function $p(.)$, the aim is to find a function $f(.)$ which provides the best recovery guarantees defined in (\ref{deltabaroriginal}).
	
	 A natural way to choose the best function $f(.)$ would be to minimize the total exponent $\bar{\psi_{tot}(p(.), f(.))}$ with respect to $f(.)$. The recovery thresholds obtained from the angle exponent based analysis used in Section \ref{sec:analysisofl1} is known to be quite accurate in practice.  For our specific case, we provide further evidence of this fact in Section \ref{sec:simulation} via simulations. However, one disadvantage
	of this method is that the dependence of the recovery threshold on the functions $f(.)$ and $p(.)$ is not very intuitive, making it difficult to find the optimal $f(.)$ for a given $p(.)$. In this section, we take an alternative approach based on estimating the Gaussian width of the cone of feasible directions of the (weighted) $\ell_1$ polytope using Gordon's theorem  \cite{Gordon87} . This method was first introduced in \cite{VR08} in the context of sparse signal recovery. It was later used by \cite{Stojnic09}, where it was used to obtain, among other quantities of interest, the 
	weak recovery threshold for the standard $\ell_1$-minimization. This threshold matched the one obtained in \cite{Donoho1} indicating that it is tight. However, to achieve this tightness, the computations used are
	quite involved and it is not clear how it may be extended to analyze weighted $\ell_1$-minimization, especially in a fairly general setting like ours. Instead, we use the analysis in \cite {Venkat12}, which sacrifices
	some tightness but leads to very amenable expressions showing the dependence on $p(.)$ and $f(.)$.
	
	Before we begin, we first introduce some notation and state a few useful results from \cite{Venkat12} that we use in our analysis. Let $\mathbf{x}$ be the underlying sparse vector normalized as before such that 
	$\| \mathbf{x}\|_{\mathbf{w},1} = 1$. Let the cone $\mathcal{C}$ be defined by $\mathcal{C}  = \{ \mathbf{v} : \| \mathbf{x} +  \mathbf{v}\|_{\mathbf{w},1} \leq  1 \}$. Let $\mathcal{S} = \mathcal{C} \cap \mathbbm{S}^{n-1}$, where
	$\mathbbm{S}^{n-1}$ is the sphere in $n$ dimensions. The Gaussian width $\omega(\mathcal{S})$ of $\mathcal{S}$ is defined as
	\begin{align*}
		\omega(\mathcal{S}) = \mathbf{E}_\mathbf{g} \left[  \sup_{\mathbf{z} \in \mathcal{S}} \mathbf{g}^T \mathbf{z} \right],
	\end{align*}
	where the expectation is taken over $\mathbf{g} \sim \mathcal{N}(0,I)$.
	The following corollary of Gordon's theorem is stated in \cite{Venkat12}.
	\begin{cor}\cite{Venkat12}	
		If $m  > \omega^2(\mathcal{S}) + 1$, then weighted $\ell_1$-minimization recovers the correct underlying sparse vector $\mathbf{x}$ with probability at least 
		$1 - \exp \left[ \frac{1}{2} (\lambda_m - \omega(\mathcal{S} ))^2\right]$. Here $\lambda_m = \sqrt{2} \frac{\Gamma \left(\frac{m+1}{2} \right)}{\Gamma \left(\frac{m}{2} \right)}$ is the expected length of a $m$-dimensional Gaussian vector.
	\end{cor}
	Since, we want to find weights which minimize the exponent in the probability of failure, in turn we should try to minimize the Gaussian width $\omega(\mathcal{S})$. 
	
	The following lemma proved in \cite{Venkat12} will help us simplify the computations needed to estimate the Gaussian width $\omega(\mathcal{S}))$.
	\begin{lem} \cite{Venkat12}
		Let $\mathcal{C}^*$ denote the polar of the cone $\mathcal{C}$. Then 
		 $\omega^2(\mathcal{S}) \leq \mathbf{E}^2_{\mathbf{g}} \left[ \mbox{dist}(\mathbf{g}, \mathcal{C}^*)\right] \leq \mathbf{E}_{\mathbf{g}} \left[ \mbox{dist}^2(\mathbf{g}, \mathcal{C}^*)\right]$.
	\end{lem}
	The first inequality in the above lemma is based on a duality argument and it is remarked in \cite{Stojnic09} that it is tight, i.e., strong duality holds. The second inequality is a direct application of Jensen's inequality. 	
	From here, we proceed to compute an upper bound on the Gaussian width along the lines of \cite{Venkat12}. The cone $\mathcal{C}^*$ is given by 
	\begin{align*}
		\mathcal{C}^* =\{ &\mathbf{z} \in \mathbbm{R}^n :  z_i = t w_i, \mbox{ for } i \in F \mbox{ and }\\
		 &|z_i| \leq t w_i, \mbox{ for } i \notin F \mbox{ for some } t > 0\}.
	\end{align*}		
	So,
	\begin{align*}
		\inf_{u \in \mathcal{C}^*}& \|\mathbf{g} - \mathbf{u} \|^2 \\
		 &= \inf_{t > 0} \sum_{ i \in F} \|g_i - tw_i \|^2 + \sum_{ i \notin F} \inf_{|u_i| \leq tw_i }\| g_i - u_i \|^2 \\
			&= \inf_{t > 0} \sum_{ i \in F} \|g_i - tw_i \|^2 + \sum_{ i \notin F} \mbox{shrink}(g_i, tw_i)^2,
	\end{align*}
	where $\mbox{shrink}(g,t)$ is the $\ell_1$ shrinkage function given by
	\begin{align*}
		\mbox{shrink}(g,t) = \begin{cases} g + t, & g < -t \\
		0, & -t \leq g \leq t \\
		g - t, & g > t.
		\end{cases}
	\end{align*}
	Hence for all $t > 0$,
	\begin{align}
		\mathbf{E}_{\mathbf{g}} & \left[ \inf_{u \in \mathcal{C}^*} \|\mathbf{g} - \mathbf{u} \|^2  \right] \\
		&\leq \sum_{ i \in F} (1 + t^2 w_i^2) + \sum_{ i \notin F} \mathbf{E} \left[  \mbox{shrink}(g_i, tw_i)^2\right]  \notag \\
		& = \sum_{ i \in F} (1 + t^2 w_i^2) + \sum_{ i \notin F} \frac{2}{\sqrt{2 \pi}} \times \\
		& \left(  -tw_i e^{-t^2w_i^2/2}  + (1 + t^2w_i^2) \int_{t w_i}^{\infty} e^{-g^2/2} dg\right) \label{approxGW}.
	\end{align}

	Recall from Section \ref{sec:withprob} that, if we divide the indices $1, \ldots, n$ into $r$ equal parts $I_1, \ldots, I_r$, then we can assume that the support of the underlying sparse vector $\mathbf{x}$ drawn according to $p(.)$ 
	has $\frac{n}{r} \bar{p}_i$ elements in $I_i$, where $\bar{p}_i = r \int_{\frac{i}{r}}^{\frac{i+1}{r}}p(u) du$. Using this and the fact that $w_i = f(i/n)$ is monotonically increasing in $i$, we get from (\ref{approxGW}),
	\begin{align*}
		\mathbf{E}_{\mathbf{g}} &\left[ \inf_{u \in \mathcal{C}^*} \|\mathbf{g} - \mathbf{u} \|^2  \right]\\
		 &\leq \sum_{i=1}^{r} \frac{n}{r} \bar{p_i} (1 + t^2 f^2( {(i+1)}/{r})) \\
		&+ \sum_{i=1}^{r}\frac{n}{r} (1 - \bar{p_i})  \frac{2}{\sqrt{2 \pi}} \left({ -t f(i/r)} e^{- t^2 f^2(i/r)/2}\right) \\
		  &+  \frac{2}{\sqrt{2 \pi}} \left((1  + t^2 f^2(i/r))\int_{tf(i/r)}^{\infty} e^{-g^2/2} dg \right)\\
		& \triangleq n L_f.
	\end{align*}
	We seek to find $f(.)$ that minimizes $L_f$. Since we are interested in the asymptotics as $n \rightarrow \infty$, we can choose $r$ large to obtain  
	\begin{align*}
		L_f &= \int_{0}^{1} p(u) (1 +   t^2f^2(u))\\
		& + (1 - p(u)) \frac{2}{\sqrt{2 \pi}} \left({ - t f(u)} e^{- t^2 f^2(u)/2} \right) +  (1 - p(u))  \\
		&\times  \frac{2}{\sqrt{2 \pi}} \left( (1  + t^2 f^2(u)) \int_{tf(u)}^{\infty} e^{-g^2/2} dg \right) du.
	\end{align*}
	Note that replacing $f(.)$ by $t f(.)$ does not change the weighted $\ell_1$-minimization and hence we can drop $t$ in the above expression. The optimal choice of $f(.)$ is obtained by 
	\begin{align}
		f(.) =\arg \min_f L_f. \label{matching}
	\end{align}
	For a given choice of $p(.)$, $L_f$ is a convex function of $f(.)$ and the above minimization can be efficiently performed numerically.
	This defines a one to one map between the probability function $p(.)$ and the corresponding weight function $f(.)$.

	The method used in this section relies on minimizing an upper bound on the error exponent, which in turn leads to a justifiable prescription for choosing the weight function $f(.)$. Once such an $f(.)$ is obtained, we can then 
	use the angle exponent based analysis of Section \ref{sec:analysisofl1} to verify the performance of this choice of $f(.)$ for the given probability function $p(.)$. 
\end{section}


\begin{section}{Numerical computations and simulations} \label{sec:simulation}
	In this section, we compute the bounds using the techniques developed in this paper for a specific probability function $p(.)$. In particular we choose $p(.)$ to be a linear function whose
	slope is governed by a parameter $c$:
\begin{align}
	p(u) = \delta -c(u-0.5), \quad u \in [0,1]. \label{probfunction}
\end{align}
The expected level of sparsity of signals drawn from this model is $\int_{0}^{1} p(u) du = \delta$.
The value of $c$ governs how tilted the signal model is. Larger values of $c$ results in a higher tilt which means a random signal drawn from this model will have most of its non-zero entries concentrated at the beginning.
To recover signals drawn from this model we use weighted $\ell_1$-minimization with choice of weights defined by the 
weight function $f(u)$ chosen to be one of the following:
\begin{itemize}
	\item[(i)]  $f(u) = 1+ \rho u, \quad u \in [0,1].$ 
	 \item[(ii)]The recommended choice of $f(.)$ given in (\ref{matching}).
\end{itemize}
The reason for choosing the linearly parameterized family of functions in item $(i)$ is twofold. One, it is a simple natural choice of a parameterized family of weights that 
may be chosen to match the probability function $p(.)$ in (\ref{probfunction}), and this provides a basis for comparing the theoretical performance of this natural choice with the recommended choice (\ref{matching}). Second, it allows us to 
demonstrate how to use the methods described in this paper for computing the total exponent to choose the ``best" weight function from a given parameterized family for a specific probability function $p(.)$.
Recall from (\ref{deltabaroriginal}) the quantity
\begin{align}
	\bar{\delta} \quad \triangleq \quad \quad \max \quad & \delta \notag \\ 
		\mbox{subject to} \quad & \bar{\psi}_{tot}(p(.; \delta, c), f(.)) \leq 0, \label{def:bardelta}
\end{align}
which is called the guaranteed bound on recoverable sparsity $\delta$.

Before proceeding to evaluation of the performance of weighted $\ell_1$-minimization in the problem specified by the above choice of functions, we will first present theoretical bounds and simulations related to the behavior
of the so-called family of leading faces $\mathcal{F}_1$, when the corresponding 
cross-polytope is described by the function $f(.)$. This is contained in Section \ref{sec:F0}.
The reader interested in the performance of weighted $\ell_1$-minimization can skip over to Section \ref{sec:withprobsim}.

\begin{subsection}{Behavior of the leading family of faces $\mathcal{F}_1$ under random projection.} \label{sec:F0}
Recall that, for a given set of weights, the weighted $\ell_1$-ball is the cross polytope in $n$ dimensions whose vertices in the first orthant are given by $\frac{e_1}{w_1}, \ldots, \frac{e_k}{w_k}$ for some $k$. We call this 
face $F_0^k$. In Section \ref{subsec:firstfamilythreshold}, we developed sufficient conditions when weighted $\ell_1$-minimization recovers a signal $\mathbf{x}$ whose support is the first $k$ indices with overwhelming probability whenever $k = \delta n$. Similar to (\ref{def:bardelta}), we can define a corresponding $\bar{\delta}$ specific to the family of leading faces given by
\begin{align*}
	\bar{\delta} \quad \triangleq \quad \quad \max \quad & \delta \\ 
		\mbox{subject to} \quad & \bar{\psi}_{tot}(\delta) \leq 0,
\end{align*}
where $\bar{\psi}_{tot}(\delta)$ is the total exponent for the face $F_0^k$ with $k=\delta n$ as defined in (\ref{eq:totalimprovedexponent}). From the above definition and Theorem \ref{thm:mainresult}, it can be concluded that
whenever $\delta < \bar{\delta}$, 
we are guaranteed to be able to recover the corresponding sparse signal with overwhelming probability. We call the quantity $\bar{\delta}$ to be the guaranteed bound on recoverable sparsity levels $\delta$.
In view of our choice of weights, which is described by the function $f(u) = 1 + \rho u$ with $\rho \geq 0$, higher values of $\rho$ correspond to more steeply varying weights. Intuitively, one may expect  that higher values of $\rho$, will make the weighted $\ell_1$ norm cost function favor non-zero entries in the first few indices which may allow the threshold $\bar{\delta}$ to be larger.  

We will show that the quantity $\bar{\delta}$ follows an increasing trend as described above. To demonstrate this for a certain choice of the parameters, we fix the compression ratio
$\frac{m}{n} = 0.5$ and compute the bound using the methods developed in Section \ref{subsec:firstfamilythreshold}. Based on Figure \ref{fig:thresholdvsm}  we choose $r=30$ as a reasonable value for the 
accuracy parameter in our computations.
Figure \ref{fig:firstfaceth} shows the dependence of this threshold on the value of $\rho$ which governs the
weight function $f(.)$.  The value of the bound at $\rho = 0$ corresponds to the case when $F_0^{\delta n}$ is a face of the regular $\ell_1$ ball. This is the threshold for $\delta$ below which standard $\ell_1$-minimization 
succeeds in recovering a signal with sparsity level $\delta$ with overwhelming probability. As expected this value matches the value reported earlier in \cite{Donoho1}.

	\begin{figure}  
		\begin{center}
		\includegraphics[scale = 0.38]{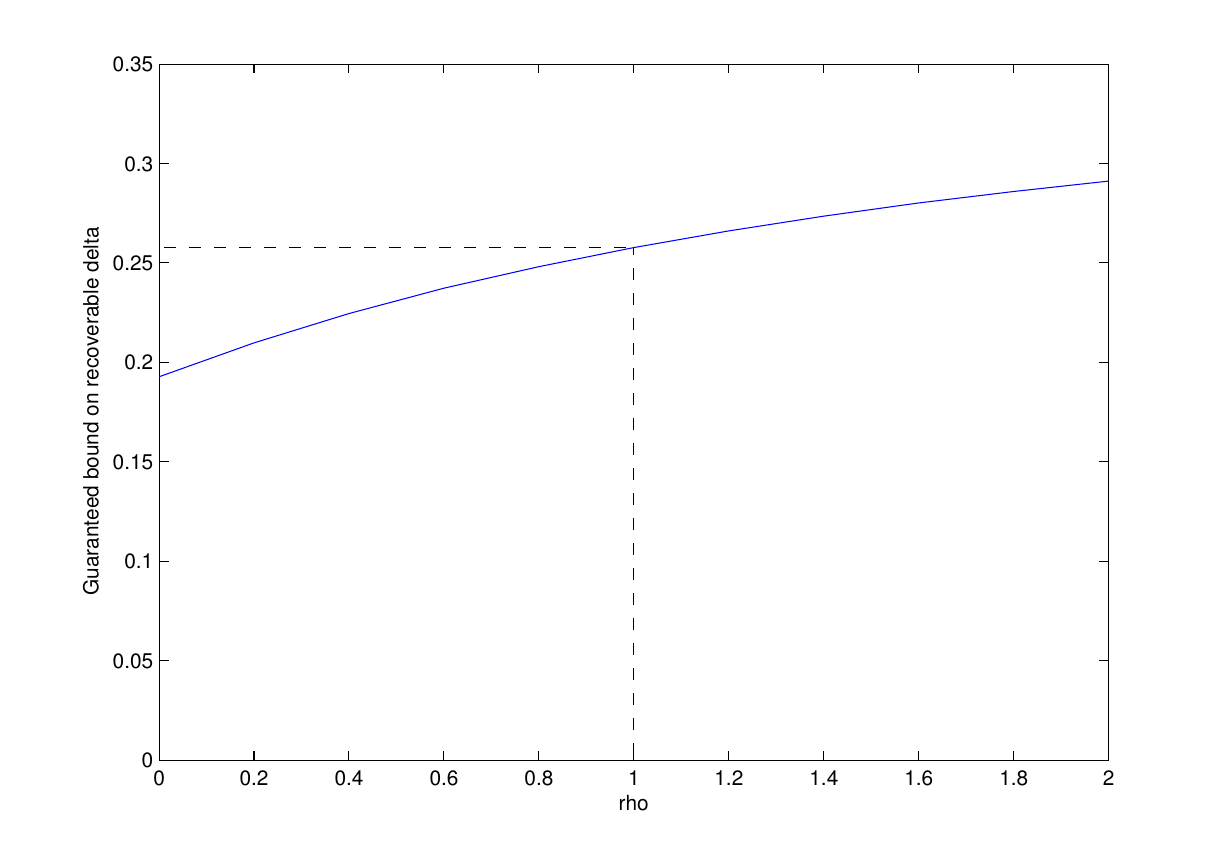}
		\end{center}
		\caption{Guaranteed bound on recoverable $\delta$ vs $\rho$ for the ``leading face" $F_0^{\delta n}$, computed
		using the methods of Section \ref{subsec:firstfamilythreshold} for $r = 30$} \label{fig:firstfaceth}
	\end{figure}

 To evaluate the accuracy of the bound, we then compare the values of the threshold predicted by the guaranteed bound to that obtained empirically through simulations for two different values of the parameter $\rho$. For this, 
we set $m = 200, n= 400$, and obtain the fraction of times $\mathbf{x} \in F_0^{\delta n}$ failed to be recovered correctly via weighted $\ell_1$-minimization from a total of 500 iterations. This is done by randomly generating a vector $\mathbf{x}$ for each iteration with support $1, \ldots, k$ and using weighted $\ell_1$-minimization to recover that $\mathbf{x}$ from its measurements given by $y = Ax$. Figure \ref{fig:firstfacesim0} and
Figure \ref{fig:firstfacesim1}  show this plot for 
$\rho = 0$ and $\rho = 1$ respectively. The vertical lines in the plots (marked A and B respectively) denote the guaranteed bounds corresponding to the value of $\rho$ in Figure \ref{fig:firstfaceth}. The simulations show a rapid fall in the
empirical value of $\mathbf{P}(E|\mathbf{x} \in F_0^{\delta n})$ around the theoretical guaranteed bound as we decrease the value of $\delta$. This indicates that the guaranteed bounds developed are fairly tight.

	\begin{figure}
		\centering
		\subfloat[nothing][$\rho  = 0$] {
		\includegraphics[scale = 0.34]{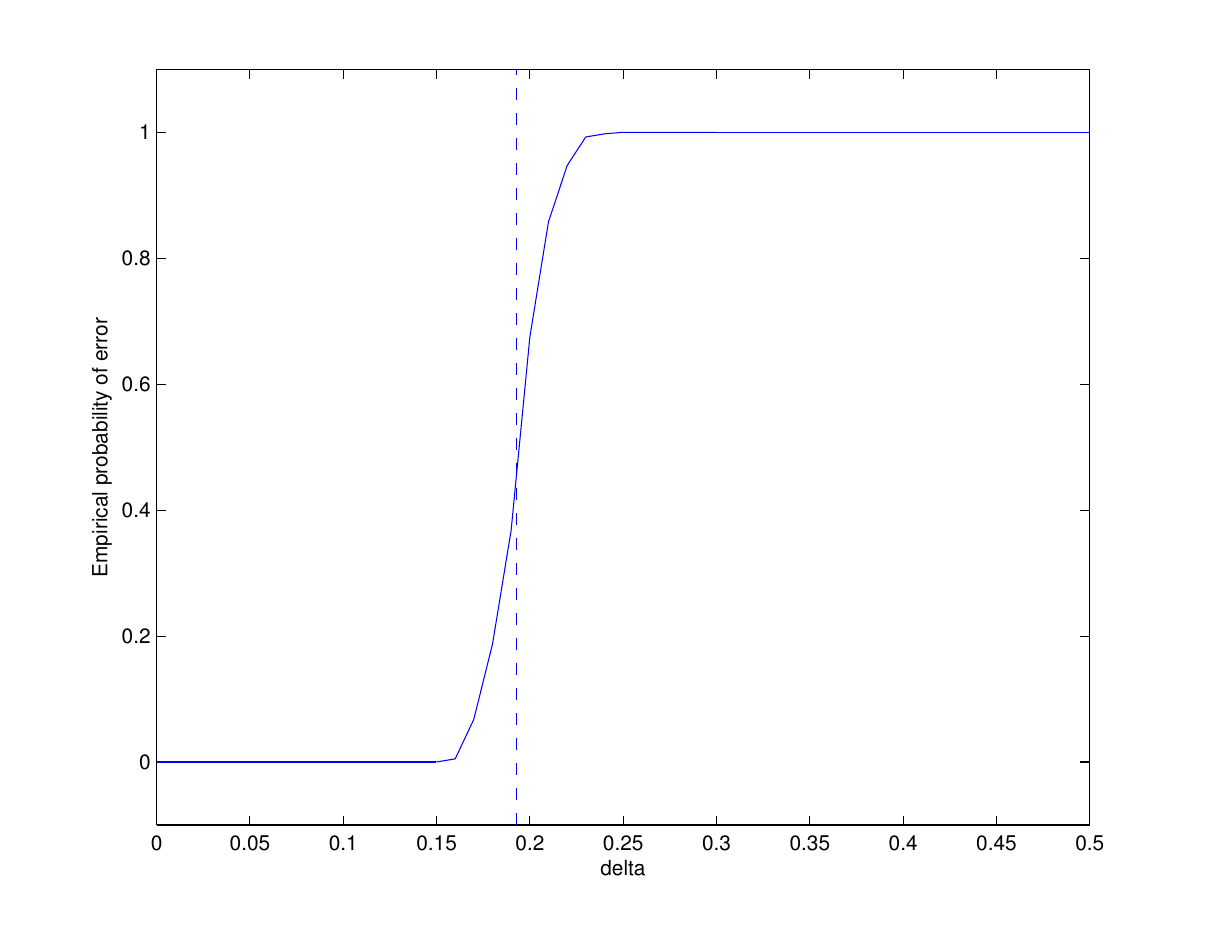}
		\label{fig:firstfacesim0} }
		\quad
		\subfloat[something][$\rho = 1$] {
		\includegraphics[scale = 0.34]{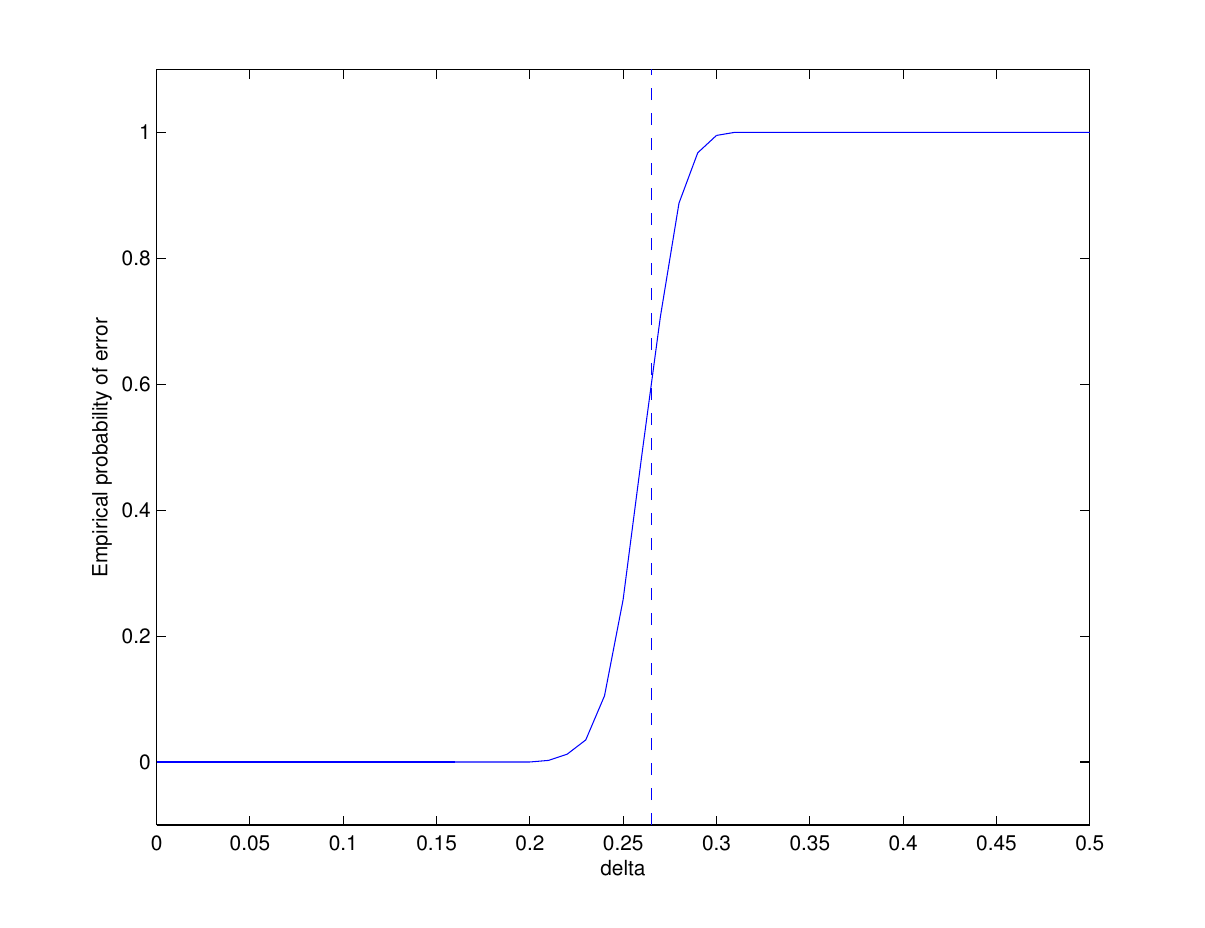}	
		\label{fig:firstfacesim1} }
		\caption{Empirical probability of error $\mathbf{P}(E|\mathbf{x} \in F_0^{\delta n})$
		 vs $\delta$ for $\rho = 0 $ and $\rho = 1$ respectively obtained through simulations using $m = 200$, $n = 400$ over $100$ experiments. The vertical line in each figure refers to
		the guaranteed bound on recoverable $\delta$ for the corresponding $\rho$ computed using methods of Section \ref{subsec:firstfamilythreshold}.}
	\end{figure}
\end{subsection}

\begin{subsection}{Performance of weighted $\ell_1$-minimization.} \label{sec:withprobsim}
 In this subsection, we compute the theoretical bound on recoverable sparsity levels using the methods  of this paper. We use the probability function in (\ref{probfunction}). We first start with the weight function $f(.)$ in item (i).
  The choice of $\rho$ plays an important role in the success of weighted $\ell_1$-minimization and it would be of interest to be able to obtain the value of $\rho$ for which one gets best performance. One way to estimate the effect of $\rho$ is to compute the guaranteed bound $\bar{\delta}$ (\ref{def:bardelta}) as suggested in Section \ref{sec:F0} and observe the trend. We can then pick the value of $\rho$
which maximizes $\bar{\delta}$.

To demonstrate this via computations we
 fix the ratio $\alpha = \frac{m}{n} = 0.5$ and compute the guaranteed bound on recoverable $\delta$ (which denotes the expected fraction of non-zero components of the signal) using the methods developed in
 Section \ref{sec:withprob}. The accuracy parameter $r$ is fixed at 60.
Figure \ref{fig:withprobth60} shows the dependence of $\bar{\delta}$ on the values of $\rho$ for three different values of $c$. 
The curves suggest that for larger values of $c$, which correspond to more rapidly decaying probabilities, the value of $\rho = \rho^{*}(c)$ which maximizes $\bar{\delta}$ is also higher. At the same time, the value of  $\bar{\delta}$ evaluated at $\rho = \rho^{*}(c)$ also increases with 
increasing $c$. This suggests that rapidly decaying probabilities allow us to recover less sparse signals by using an appropriate weighted 
$\ell_1$-minimization. 


	
	\begin{figure}
		\begin{center}
		\includegraphics[scale = 0.37]{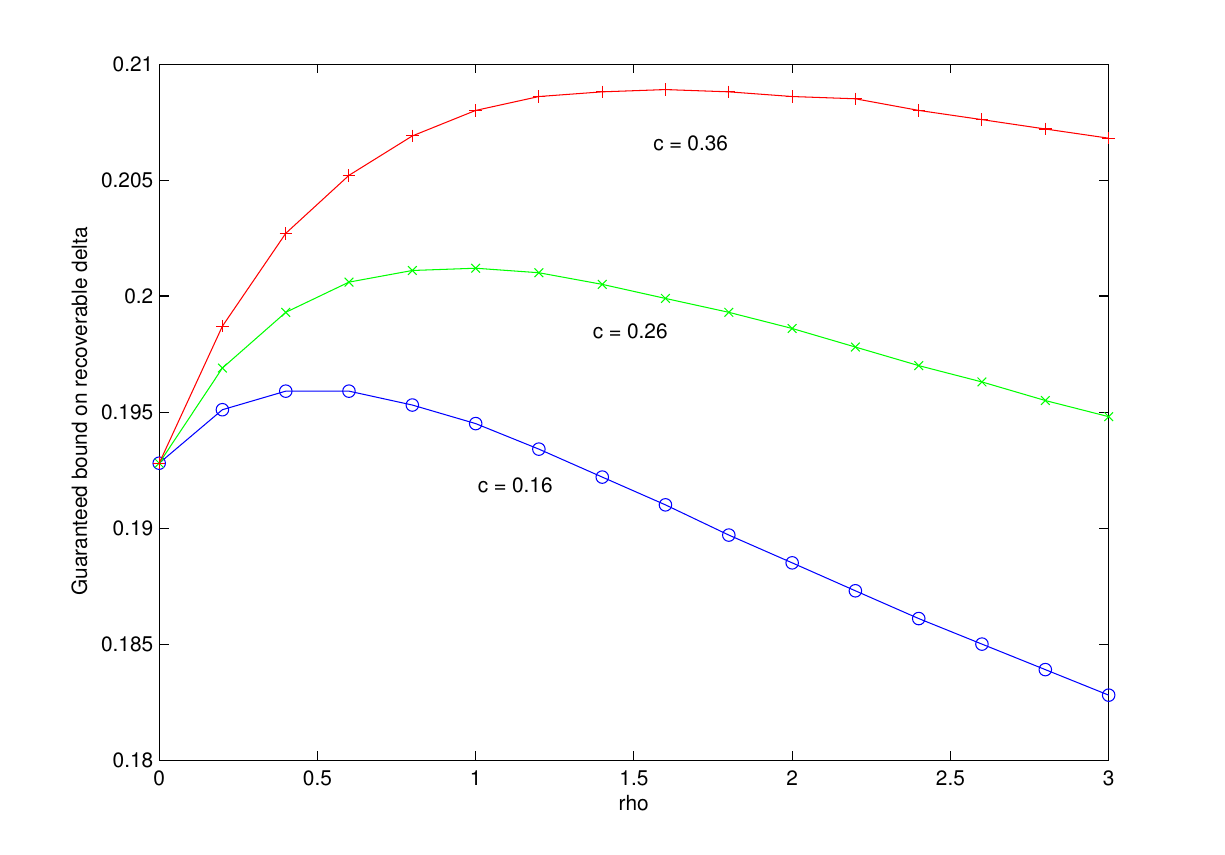}
		\end{center}
		\caption{Guaranteed bound on recoverable $\delta$ vs $\rho$ for $\frac{m}{n}=0.5$ computed using the methods of this paper, for $c=0.16$, $c=0.26$ and $c=0.36$. The parameter $r$ is fixed at 60.}
		\label{fig:withprobth60}
	\end{figure}
	
	\begin{figure}
		\begin{center}
		\includegraphics[scale = 0.37]{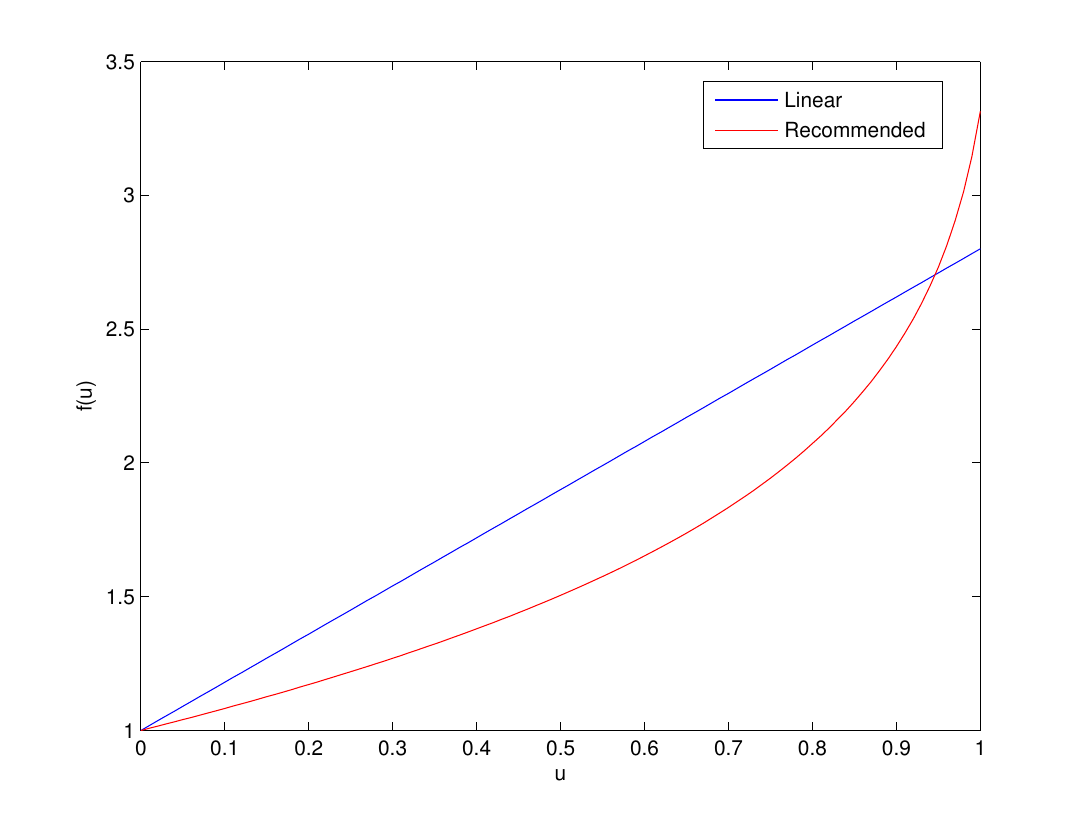}
		\end{center}
		\caption{Comparison of the weight function $f(.)$ obtained by finding the best among the linearly parameterized family and the recommended $f(.)$ from (\ref{matching}) for $\delta = 0.19$, $c = 0.36$.}
		\label{fig:linVsrec}
	\end{figure}

Second, we compute the recommended choice of weight function by numerically solving for $f(.)$ in (\ref{matching}) at the required evaluation points. We then compute the guaranteed bound $\bar{\delta}$ defined in (\ref{def:bardelta}) by using this choice of weights. The parameters 
$\alpha$ and $r$ are fixed at $0.5$ and $60$ as in the previous case above.

Comparing the values of $\bar{\delta}$ from Table \ref{tbl:rhostar} and Table \ref{tbl:bardelta}, we can see that using the recommended choice of weight function $f(.)$ does indeed have a larger guaranteed bound on recoverable sparsity for each
of the choices of $p(.)$ considered. 
	
To provide evidence that weighted $\ell_1$-minimization indeed improves performance, we conduct simulations to compare standard and weighted $\ell_1$-minimization. We fix the value of $\delta$ to be $0.185$. We then explore the effect of choosing different values of the model parameter $c$ in (\ref{probfunction}) on recoverability. We sample random signals with supports generated by the distribution imposed by $p(.)$. We then choose the weight function $f(.)$ in two different ways. One by utilizing the curves computed in Figure \ref{fig:withprobth60} to make the best choice of $\rho$ (see Table \ref{tbl:rhostar}) and second by using the recommended choice from (\ref{matching}). We use weighted $\ell_1$-minimization corresponding to these choice of weights to recover the generated signal from its measurements. We compute the fraction of the experiments for which this method fails to recover the correct signal over
500 experiments. The values of $m$ and $n$ are chosen to be $500$ and $1000$ respectively. To compare the performance of weighted $\ell_1$-minimization to standard $\ell_1$-minimization, we repeat the same procedure but use standard $\ell_1$-minimization to recover the signal. Figure \ref{fig:withprobsimPfvsC} compares the values generated by each method. Notice how the performance of the standard $\ell_1$-minimization method remains more or less invariant with increasing $c$. This shows that standard $\ell_1$-minimization fails to exploit the extra information present because of the knowledge of $c$ (i.e. the decaying nature of the probabilities) and its performance depends only on the value of $\delta$, the expected fractional level of sparsity and is insensitive to the tilt of the model given by $c$. On the other hand, the performance of weighted $\ell_1$-minimization improves with $c$ for both choice of weight functions, with the recommended choice (\ref{matching}) showing the best performance.

\begin{table}
	\caption{$c$ vs $\rho^{*}(c)$ using theoretical guaranteed bounds with $r=60$ (Figure \ref{fig:withprobth60}) }.
	\label{tbl:rhostar}
	\begin{center}
\begin{tabular}{| c || c | c | c | c |}   %
\hline
	c & 0   & 0.16 & 0.26 & 0.36 \\
\hline	
	$\rho$ &  0  & 0.6 & 1.0 & 1.6 \\
\hline
	$\bar{\delta}$ & 0.1928&  0.1959& 0.2012 & 0.2089 \\
\hline	
\end{tabular}
\end{center}

\end{table}

\begin{table}
	\caption{$\bar{\delta}$ vs $c$ for recommended choice of $f(.)$ in (\ref{matching})}.
	\label{tbl:bardelta}
	\begin{center}
\begin{tabular}{| c || c | c | c | c |}  %
\hline
	c & 0   & 0.16 & 0.26 & 0.36 \\
\hline	
	$\bar{\delta}$ & 0.1928& 0.1960  & 0.2014  & 0.2097 \\
\hline
\end{tabular}
\end{center}

\end{table}


	\begin{figure}
		\begin{center}
		\includegraphics[scale = 0.38]{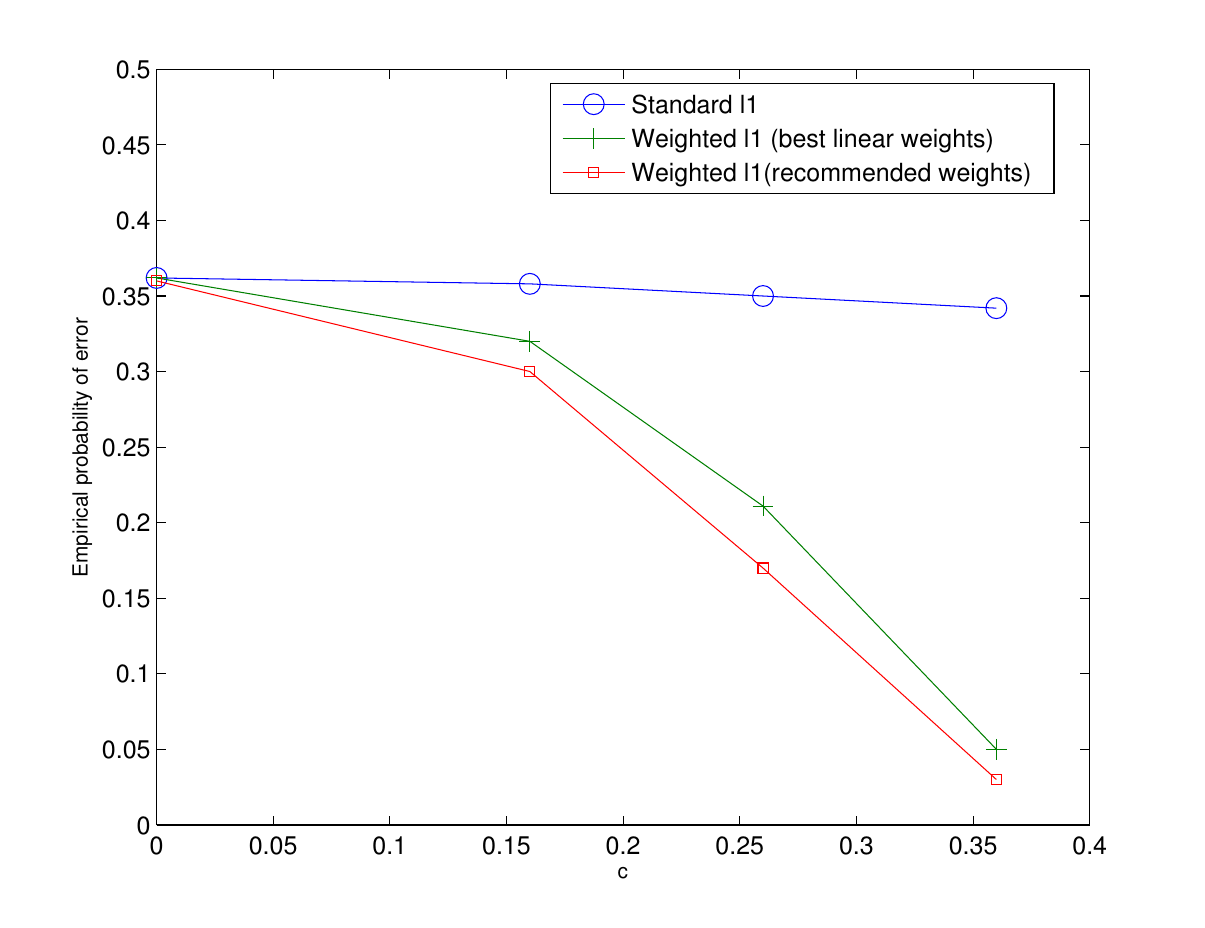}
		\end{center}
		\caption{Empirical probability of error $\mathbf{P}(E)$ of weighted $\ell_1$-minimization vs the tilt of the model given by the parameter $c$. Probability function  $p(u) = 0.185 - c(u-0.5)$ and weight
		function (i) (plus) $f(u) = 1 + \rho^*(c)$  where $\rho^*(c)$ is the optimal value of $\rho$ obtained form Figure \ref{fig:withprobth60} and (ii) (square) $f(u)$ chosen according to (\ref{matching}). Problem size is given by $m=500$, $n=1000$. Number of experiments $= 500$. }
		\label{fig:withprobsimPfvsC}
	\end{figure}

\end{subsection}
\end{section}

\section{Conclusion} \label{sec:conclusion}
In this paper we analyzed sparse signal recovery via weighted $\ell_1$-minimization for a special class of probabilistic signal model, namely when the 
weights are uniform samples of a continuous function. We leveraged the techniques developed in 
\cite{Donoho1} and \cite{WeiyuThesis} to provide sufficient conditions under which weighted $\ell_1$-minimization succeeds in recovering the sparse signal with overwhelming probability. 
In the process, we also provided conditions under which certain 
special class of faces of the skewed cross-polytope get ``swallowed" under random projections.

A question central to the weighted $\ell_1$-minimization based approach is the optimal choice of weights. In this paper, we provide a general framework to choose these weights. The recovery capacity of this choice of weights can also be verified theoretically by computing the angle exponents using the methods of this paper.

\appendices
\section{proof of Lemmas \ref{lem:basic_exponent_monotonicity} and \ref{lem:internalexponentcurvature}}
\subsection{Proof of Lemma \ref{lem:basic_exponent_monotonicity}}
Let $G_0$ be the face whose vertices are given by $\frac{1}{w_1}e_1, \ldots, \frac{1}{w_k}e_k, \frac{1}{w_{k+1}}e_{k+1}, \ldots , \frac{1}{w_l}e_l$. Let $G$ be any face other than $G_0$ whose vertices are given by 
$\frac{1}{w_1}e_1, \ldots, \frac{1}{w_k}e_k, \frac{1}{w_{n_{k+1}}}e_{n_{k+1}}, \ldots , \frac{1}{w_{n_l}}e_{n_l}$. Consider forming a sequence of faces $G^0, G^1, \ldots, G^{l-k}$, where $G^0 = G$, and $G^{i+1}$ is
obtained from $G^i$ by swapping the vertices $\frac{1}{w_{n_{l-i+1}}} e_{n_{l-i+1}}$ and $\frac{1}{w_{l-i+1}} e_{l-i+1}$. Since $w_{n_{l-i+1}} \geq w_{l-i+1}$, the expression in Lemma \ref{lem:externalexpression} for
the external angle increases at each step. Hence, $\gamma(G_0,P) \geq \gamma(G,P)$.

For a fixed value of $l$, the exponent for the internal angle is only affected by the term $p_Z(0)$ in the expression for internal angle exponent in Lemma \ref{lem:internalexpression}. Also 
\begin{align*}
	p_Z(0) = 2 \int_0^{\infty} v p_{Y_0}(v) F_S(v) dv.
\end{align*}
Following the same procedure as above for generating the sequence of faces $G^i$, it can be seen that at each step the variance of some $Y_p$ is decreased while keeping the other $Y_i$s unchanged. Thus 
$F_S(v)$ in the above expression for $p_Z(0)$ increases at each step. Thus, the exponent for $\beta(F_0,G_0)$ is greater than or equal to that of $\beta(F_0,G)$.
\subsection{Proof of Lemma \ref{lem:internalexponentcurvature}}
We rewrite the expression for the combinatorial exponent from Section \ref{sec:firstfacecombinatorialexponent}:
	\begin{align*}
		\psi_{com}(\mathbf{h}) &= \frac{1-\delta}{r}  \sum_{i=1}^{r} H \left(   \frac{ h_i}{1-\delta} \right) \\
		&+ (\sum_{i=1}^{r} r h_i - \delta) \log 2.
	\end{align*}	
	The concavity of this function follows from the fact that the standard entropy function $H(.)$ is a concave function.
	From its expression in Section \ref{sec:firstfaceexternalangle} we observe that the external angle exponent is a linear function of $\mathbf{h}$ and hence a concave function.
	The concavity of the internal angle exponent is slightly more involved and we spend the rest of this section in proving it.
	
	The internal angle exponent as computed in Section \ref{sec:firstfaceinternalangle} is given by 
	\begin{align*}
		\psi_{int}(\mathbf{h})  &= -(\sum_{i=1}^{r}r h_i  - \delta) \log 2 \\
		&- \left(  \frac{\bar{c}_0^2}{2c_1}y^2 + \bar{c}_0 \lambda_0^{*}(y)  \right).
	\end{align*}
	The quantity $\bar{c}_0 = \bar{c}_0(\mathbf{h}) = \sum_{i} f_i h_i$ is a linear function of $\mathbf{h}$. So $\frac{\bar{c}_0(\mathbf{h})^2}{2c_1}y^2$ is a convex quadratic function of $\mathbf{h}$. 
	Therefore it suffices to prove that
	$F(\mathbf{h}) \triangleq \bar{c}_0(\mathbf{h}) \lambda_0^{*}(y) $ is convex in $\mathbf{h}$.
	Recall that
	\begin{align*}
		\lambda_0^{*}(y)  = \lambda_0^{*}(y,\mathbf{h})  =  \max_{s} sy - \frac{1}{r \bar{c}_0(\mathbf{h})} \sum_{i=1}^{r} h_i \lambda(s f_i).
	\end{align*}
	Therefore
	\begin{align*}
		F(\mathbf{h}) = \max_{s} \bar{c}_0(\mathbf{h}) sy - \frac{1}{r} \sum_{i=1}^{r} h_i \lambda(s f_i).
	\end{align*}
	Since the argument in the above maximization is linear in $\mathbf{h}$ it follows that $F(\mathbf{h})$ is convex in $\mathbf{h}$.

\section{Angle exponents for the typical face $F_p$}
Divide the interval $[0,1]$ into $r$ equally spaced intervals. Let the face $F$ in consideration have $n\delta_i$ indices in the $i^{th}$ interval. Also, let $g_i = r\delta_i$. The asymptotic exponents for this face can be
obtained easily by a straight-forward generalization of the procedure described in Section \ref{sec:firstface}.
We give the final expressions for the combinatorial, internal and external angle exponents for a given value of $g = (g_1, g_2, \ldots, g_r)^T$. In what follows, we use $f_i = f\left( \frac{i}{r} \right)$.
		\begin{subsubsection}{Combinatorial Exponent}
			\begin{align*}
				\psi_{com}(\mathbf{h}) &= \frac{1}{r} \sum_{i=1}^{r} (1 - g_i) H \left(  \frac{h_i}{1-g_i} \right) \\
				&+ r \sum_{i=1}^{r} (h_i - g_i),
			\end{align*}	
			where $H(.)$ is the binary entropy function with base $e$.
		\end{subsubsection}

		\begin{subsubsection}{Internal Angle Exponent}
			The negative of the internal angle exponent is given by
			\begin{align*}
				\psi_{int}(\mathbf{h},y) = - r \sum_{i=1}^{r} (h_i - g_i) - \left( \frac{c_0^2}{2c_1}y^2 +c_0 \lambda_0^{*}(y) \right),
			\end{align*}
			where
			\begin{align*}
				&c_0 = \frac{1}{r} \sum_{i=1}^{r} f_i h_i, \\
				&c_1 = \frac{1}{r} \sum_{i=1}^{r} f_i^2 g_i, \\
				&\lambda_0^{*}(y) = \max_{s} sy - \lambda_0(s), \\
				&\lambda_0(s) = \frac{1}{r} \sum_{i=1}^{r} \frac{1}{c_0} \lambda(sf_i) (h_i).
			\end{align*}	
			Here $\lambda(u) = \frac{u^2}{2} + \log (2 \Phi(u))$ is the characteristic function of the standard half-normal distribution.
		\end{subsubsection}
		
		\begin{subsubsection}{External Angle Exponent}
		The negative of the external angle exponent is given by
			\begin{align*}
				\psi_{ext}(\mathbf{h},x) = -\left( c_2x^2 - \log(G_0(x)) \right),
			\end{align*}
			where,
			\begin{align*}
				c_2 &=  \frac{1}{r} \sum_{i=1}^{r} f_i^2 (g_i + h_i), \ \mbox{and} \\
				\log(G_0(x)) &= \int_{0}^{1}  \log(\erf (xf(u)))du \\
				&-  \frac{1}{r} \sum_{i=1}^{r}  \log(\erf (xf_i)) (h_i + g_i).
			\end{align*}
		\end{subsubsection}
		
		\begin{subsubsection}{Total Exponent} \label{completetotalexponent}
			Combining the exponents we define the total exponent as
			\begin{align*}
				\psi_{tot} = \max_{\mathbf{h},x,y} \quad &\psi_{com} + \psi_{int}(\mathbf{h},y) + \psi_{ext}(\mathbf{h},x) \\
				\mbox{subject to} \quad &\frac{1}{r} \sum_{i=1}^{r} h_i \geq \alpha - \delta, \\
								      &0 \leq h_i \leq 1 - g_i,
			\end{align*}
			where $\delta = \frac{1}{r} \sum_{i=1}^{r}g_i$. From Theorem \ref{thm:firstfacetotalexponent}, the total exponent satisfies
			\begin{align*}
				\frac{1}{n} \log (\mathbf{P}(E|\mathbf{x} \in F))  \leq  \bar{\psi}_{tot} + o(1).
			\end{align*}
			So as long as the quantity $\bar{\Psi}_{tot} < 0$, weighted $\ell_1$-minimization succeeds in recovering the sparse signal with an exponentially small probability of failure.
		\end{subsubsection}



\ifCLASSOPTIONcaptionsoff
  \newpage
\fi

\bibliographystyle{IEEEtran}
\bibliography{IEEEabrv,References}

\begin{thebibliography}{10}
\providecommand{\url}[1]{#1}
\csname url@samestyle\endcsname
\providecommand{\newblock}{\relax}
\providecommand{\bibinfo}[2]{#2}
\providecommand{\BIBentrySTDinterwordspacing}{\spaceskip=0pt\relax}
\providecommand{\BIBentryALTinterwordstretchfactor}{4}
\providecommand{\BIBentryALTinterwordspacing}{\spaceskip=\fontdimen2\font plus
\BIBentryALTinterwordstretchfactor\fontdimen3\font minus
  \fontdimen4\font\relax}
\providecommand{\BIBforeignlanguage}[2]{{%
\expandafter\ifx\csname l@#1\endcsname\relax
\typeout{** WARNING: IEEEtran.bst: No hyphenation pattern has been}%
\typeout{** loaded for the language `#1'. Using the pattern for}%
\typeout{** the default language instead.}%
\else
\language=\csname l@#1\endcsname
\fi
#2}}
\providecommand{\BIBdecl}{\relax}
\BIBdecl

\bibitem{CandesTao}
E.~Cand\`{e}s and T.~Tao, ``Decoding by linear programming,'' \emph{IEEE Trans.
  on Information Theory}, vol.~51, no.~12, pp. 4203--4215, December 2005.

\bibitem{Donoho1}
D.~Donoho, ``High dimensional centrally symmetric polytopes with neighborliness
  proportional to dimension,'' \emph{Discrete and Computational Geometry}, vol.
  102, no.~27, pp. 617--652, 2006.

\bibitem{NeedellTropp}
D.~Needell and J.~A. Tropp, ``Co{S}a{M}{P}: Iterative signal recovery from
  incomplete and inaccurate samples,'' \emph{Appl. Comput. Harmon. Anal.},
  vol.~26, no.~3, pp. 301--321, 2009.

\bibitem{BlumensathDavis}
T.~Blumensath and M.~E. Davis, ``Iterative hard thresholding for compressed
  sensing,'' \emph{Applied and Computational Harmonic Analysis}, vol.~27,
  no.~3, pp. 265--274, 2009.

\bibitem{RechtFazelParrilo10}
B.~Recht, M.~Fazel, and P.~Parrilo, ``Guaranteed minimum-rank solutions of
  linear matrix equations via nuclear norm minimization,'' \emph{SIAM Review},
  vol.~52, no.~3, pp. 471--501, 2010.

\bibitem{RechtXuHassibi10}
B.~Recht, W.~Xu, and B.~Hassibi, ``Necessary and sufficient conditions for
  success of the nuclear norm heuristic for rank minimization,''
  \emph{Mathematical Programming, Series B}, vol. 125, no. DOI
  10.1007/s10107-010-0422-2, 2010.

\bibitem{BaraniukRosing07}
W.~Dai, M.~Seikh, O.~Milenkovic, and R.~Baraniuk, ``Compressive sensing {DNA}
  {M}icroarrays,'' \emph{EURASIP Journal on Bioinformatics and Systems
  Biology}, 2009.

\bibitem{Sidhant08}
H.Vikalo, F.~Parvaresh, S.Misra, and B.Hassibi, ``Recovering sparse signals
  using sparse measurement matrices in compressed {DNA} microarrays,''
  \emph{IEEE Trans. on Signal Processing, special issue on genomic signal
  processing}, vol.~2, no.~3, June 2008.

\bibitem{Baraniuk1}
R.~G. Baraniuk, V.~Cevher, M.~F. Duarte, and C.~Hegde, ``Model based
  compressive sensing,'' \emph{IEEE Trans. on Information Theory}, vol.~56, pp.
  1982--2001, April 2010.

\bibitem{WeiyuThesis}
W.~Xu, ``Compressive sensing for sparse approximations: Construction,
  algorithms and analysis,'' Ph.D. dissertation, California {I}nstitute of
  {T}echnology, Department of Electrical Engineering, 2010.

\bibitem{XuHassibi10}
M.~A. Khajehnejad, W.~Xu, A.~S. Avestimehr, and B.~Hassibi, ``Analyzing
  weighted $\ell_1$ minimizaion for sparse recovery with nonuniform sparse
  models,'' \emph{IEEE Trans. on Signal Processing}, vol.~59, no.~5, pp.
  1985--2001, 2011.

\bibitem{DonohoTanner05}
D.~Donoho and J.~Tanner, ``Neighborliness of randomly-projected simplices in
  high dimensions,'' \emph{Proc. National Academy of Sciences}, vol. 102,
  no.~27, 2005.

\bibitem{XuHassibi08}
W.~Xu and B.~Hassibi, ``Compressed sensing over the {G}rassmann manifold: A
  unified analytical framework,'' \emph{Allerton Conference}, 2008.

\bibitem{San52}
L.~A. Santal\'{o}, ``Geometr\'{i}a integral en espacios de curvatura
  constante.'' \emph{Rep. Argentina Publ. Com. Nac. Energ\'{i}a At\'{o}mica,
  Ser. Mat 1, No. 1}, 1952.

\bibitem{McM75}
P.~McMullen, ``Non-linear angle-sum relations for polyhedral cones and
  polytopes,'' \emph{Math. Proc. Cambridge Philos. Soc.}, vol.~78, no.~2, pp.
  247--261, 1975.

\bibitem{Gordon87}
Y.~Gordon, ``On milman's inequality and random subspaces with escape through a
  mesh in $\mathbbm{R}^n$,'' \emph{Geometric aspects of functional analysis,
  Israel Seminar 1986 - 87, Lecture Notes in Mathematics}, vol. 1317, pp.
  84--106.

\bibitem{VR08}
M.~Rudelson and R.~Vershynin, ``On sparse reconstruction from {F}ourier and
  {G}aussian measurements,'' \emph{Communications on Pure and Applied
  Mathematics}, vol.~61, no.~8, pp. 1025--1045.

\bibitem{Stojnic09}
M.~Stojnic, ``Various thresholds for $\ell_1$ optimization in comressed
  sensing,'' \emph{arXiv}, no. 0907.3666v1, 2009.

\bibitem{Venkat12}
V.~Chandrasekaran, B.~Recht, P.~Parrilo, and A.~Willsky, ``The convex geometry
  of linear inverse problems,'' \emph{Foundations of Computational
  Mathematics}, vol.~12, no.~6, pp. 805--849, 2012.

\end{thebibliography}

\end{document}